\documentclass{www-style}

\usepackage{amsmath}
\usepackage{amssymb}
\usepackage{stmaryrd}
\usepackage{times}
\usepackage{ifsym}
\usepackage{color}
\usepackage{theorem}
\usepackage{textcomp}

\newcommand{\Rel} {\text{\sc Rel}}
\newcommand{\Attr} {\text{\sc Attr}}
\newcommand{\PK} {\text{\sc PK}}
\newcommand{\FK} {\text{\sc FK}}
\newcommand{\Value} {\text{\sc Value}}

\newcommand{\Class} {\text{\sc Class}}
\newcommand{\ObjP} {\text{\sc OP}}
\newcommand{\DTP} {\text{\sc DTP}}
\newcommand{\concat} {\text{\sc Concat}}
\newcommand{\gClassIRI} {\text{\sc ClassIRI}}
\newcommand{\gObjPIRI} {\text{\sc OP\_IRI}}
\newcommand{\gID} {\text{\sc TupleID}}

\newcommand{\gDTPIRI} {\text{\sc DTP\_IRI}}
\newcommand{\gRIRI} {\text{\sc RowIRI}}
\newcommand{\gRBN} {\text{\sc BlankNode}}

\newcommand{\ExistThreeAttr} {\text{\sc ThreeAttr}}
\newcommand{\TwoFKfrom} {\text{\sc TwoFK}}
\newcommand{\FKfrom} {\text{\sc OneFK}}
\newcommand{\FKto} {\text{\sc FKto}}
\newcommand{\HasPK} {\text{\sc HasPK}}
\newcommand{\IsValue} {\text{\sc IsValue}}
\newcommand{\Violation} {\text{\sc Violation}}

\newcommand{\BinRel} {\text{\sc BinRel}}
\newcommand{\IsBinRel} {\text{\sc IsBinRel}}
\newcommand{\Triple} {\text{\sc Triple}}

\newcommand{\rdftype} {\text{\tt "rdf:type"}}
\newcommand{\owlClass} {\text{\tt "owl:Class"}}
\newcommand{\owlOP} {\text{\tt "owl:ObjectProperty"}}
\newcommand{\domain} {\text{\tt "rdfs:domain"}}
\newcommand{\range} {\text{\tt "rdfs:range"}}
\newcommand{\owlDTP} {\text{\tt "owl:DatatypeProperty"}}
\newcommand{\owldf} {\text{\tt "owl:differentFrom"}}

\newcommand{\fk}{\subseteq_\text{FK}}

\newcommand{\R}{{\bf R}}
\newcommand{\nil}{\text{\tt NULL}}
\newcommand{\D}{{\bf D}}
\newcommand{\att}{{\it att}}

\newcommand{\bag}[1]{\ensuremath{\left\{\kern-.3em\left\{ #1 \right\}\kern-.3em\right\}}}

\newcommand{\ignore}[1]{}

\newcommand{\NAME}{\text{\tt "NAME"}}

\newcommand{\SID}{\text{\tt "SID"}}
\newcommand{\CID}{\text{\tt "CID"}}
\newcommand{\DID}{\text{\tt "DID"}}
\newcommand{\CODE}{\text{\tt "CODE"}}
\newcommand{\STUDENT}{\text{\tt "STUDENT"}}
\newcommand{\DEPT}{\text{\tt "DEPT"}}

\newcommand{\COURSE}{\text{\tt "COURSE"}}

\newcommand{\ENROLLED}{\text{\tt "ENROLLED"}}

\newcommand{\DM}{{\cal DM}}

\newcommand{\base}{\text{\tt base}}

\newcommand{\V}{{\bf V}}
\newcommand{\I}{{\bf I}}
\newcommand{\B}{{\bf B}}
\renewcommand{\L}{{\bf L}}
\newcommand{\M}{{\cal M}}
\newcommand{\N}{{\cal N}}
\newcommand{\rdf}{{\cal G}}
\newcommand{\rel}{{\cal RC}}
\newcommand{\ins}{{\cal I}}
\newcommand{\tr}{\text{\it tr}}

\newcommand{\AND}{\operatorname{AND}}
\newcommand{\OPT}{\operatorname{OPT}}
\newcommand{\UNION}{\operatorname{UNION}}
\newcommand{\FILTER}{\operatorname{FILTER}}
\newcommand{\SELECT}{\operatorname{SELECT}}
\newcommand{\AS}{\operatorname{AS}}
\newcommand{\uni}{~{\operatorname{UNION}}~}
\newcommand{\opt}{~{\operatorname{OPT}}~}
\newcommand{\andp}{~{\operatorname{AND}}~}
\newcommand{\vc}{~{\operatorname{FILTER}}~}
\newcommand{\bound}{\operatorname{bound}}
\newcommand{\sel}[2]{(\SELECT \ #1 \ #2)}
\newcommand{\bw}[1]{{\normalsize #1}}

\newtheorem{definition}{Definition}
\newtheorem{example}{Example}
\newtheorem{theorem}{Theorem}
\newtheorem{proposition}{Proposition}
\newtheorem{lemma}{Lemma}

\newcommand{\isn}{\text{\tt IsNull}}
\newcommand{\isnn}{\text{\tt IsNotNull}}

\newcommand{\nj}{\bowtie}

\newcommand{\semr}[1]{\llbracket #1 \rrbracket_I}
\newcommand{\MINUS}{\operatorname{MINUS}}
\newcommand{\egp}{\{\ \}}
\newcommand{\minus}{~{\operatorname{MINUS}}~}
\newcommand{\as}{~{\operatorname{AS}}~}
\newcommand{\isBlank}{\operatorname{isBlank}}
\newcommand{\isIRI}{\operatorname{isIRI}}
\newcommand{\var}{\operatorname{var}}
\newcommand{\dom}{\operatorname{dom}}
\newcommand{\res}[2]{#1_{|_{#2}}}
\newcommand{\ren}[1]{\rho_{\{#1\}}}
\newcommand{\sem}[1]{\llbracket #1 \rrbracket_G}
\newcommand{\semp}[2]{\llbracket #1 \rrbracket_{#2}}
\newcommand{\Qone}{{Q_1}}
\newcommand{\Qtwo}{{Q_2}}
\newcommand{\Qstarone}{{Q^\star_1}}
\newcommand{\Qstartwo}{{Q^\star_2}}
\newcommand{\X}{{\cal X}}

\begin{document}
\clubpenalty=10000 
\widowpenalty = 10000

\newenvironment{alfl}
{
\begin{list}
{}
{
\setlength{\topsep}{1pt}
\setlength{\itemsep}{0pt}
\settowidth{\labelwidth}{$\bullet$}
\setlength{\leftmargin}{13pt} 
\setlength{\itemindent}{0pt}
}
}
{
\end{list}
}

\newenvironment{alsl}
{
\begin{list}
  {}
  {
  \setlength{\topsep}{0pt}
  \setlength{\itemsep}{0pt}
  \settowidth{\labelwidth}{(b)}
  \setlength{\leftmargin}{10pt} 
  \setlength{\itemindent}{0pt}
  }
}
{
\end{list}
}

\setlength\theorempreskipamount{5pt plus 1pt minus 1pt}
\setlength\theorempostskipamount{5pt plus 1pt minus 1pt}

\title{On Directly Mapping Relational Databases to RDF and OWL (Extended Version)}
%
%
%
%
%

\numberofauthors{3} 
%
\author{
%
%
\alignauthor
Juan F. Sequeda\\
       \affaddr{University of Texas at Austin}
       \email{jsequeda@cs.utexas.edu}
\alignauthor
Marcelo Arenas\\
       \affaddr{PUC Chile}
       \email{marenas@ing.puc.cl}
\alignauthor
Daniel P. Miranker\\
       \affaddr{University of Texas at Austin}
       \email{miranker@cs.utexas.edu}
}
\date{}

\maketitle
\begin{abstract}
Mapping relational databases to RDF is a fundamental problem for the
development of the Semantic Web.  We present a solution, inspired by
draft methods defined by the W3C where relational databases are
directly mapped to RDF and OWL.  Given a relational database schema
and its integrity constraints, this direct mapping produces an OWL
ontology, which, provides the basis for generating RDF
instances.  The semantics of this mapping is defined using Datalog.
Two fundamental properties 
are information
preservation and query preservation. We prove that our mapping
satisfies both conditions, even for
relational databases that contain null values. We also consider two
desirable properties: 
monotonicity and semantics
preservation. We prove that our mapping is monotone and also prove
that no monotone mapping, including ours, is semantic preserving.  We
realize that monotonicity is an obstacle for semantic preservation and
thus present a non-monotone direct mapping that is semantics
preserving.
\end{abstract}

\category{H.2.5}{Heterogeneous Databases}{Data translation}
\category{H.3.5}{Online Information Services}{Web-based services}


\keywords{Relational Databases, Semantic Web, Direct Mapping, RDB2RDF, SQL, SPARQL, RDF, OWL}

\section{Introduction}
\noindent
In this paper, we study the problem of directly mapping a relational
database
to an RDF graph with OWL vocabulary.
A direct mapping is a default and automatic way of
translating a relational database to RDF. One report suggests that
Internet accessible databases contained up to 500 times more data
compared to the static Web and roughly 70\% of websites are backed by
relational databases, making automatic translation of relational
database to RDF central to the success of the Semantic Web~\cite{BMZ+07}.

We build on an existing direct mapping of
relational database schema to OWL DL~\cite{TSM08} and the current draft of
the W3C Direct Mapping standard \cite{APS11}. We
study two properties that are fundamental to a direct mapping:
information preservation and query preservation. Additionally we study two desirable properties: monotonicity and semantics preservation. To the best of our knowledge, we are presenting the first direct mapping from a relational database to an RDF graph with OWL vocabulary that has been thoroughly studied with respect to these fundamental and desirable properties.    

Information preservation speaks to the ability of
reconstructing the original database from the result of the direct
mapping. Query preservation means that every query over a relational
database can be translated into an equivalent query over the result of
the direct mapping. Monotonicity is a desired property because it assures that a
re-computation of the entire mapping is not needed after any updates
to the database.  Finally, a direct mapping is semantics preserving
if the satisfaction of a set of integrity constraints 
are encoded in the mapping result. 

Our proposed
direct mapping is monotone, information preserving and query preserving even in the general and practical scenario where
relational databases contain null values. However, given a database that violates an integrity constraint, our direct mapping
generates a consistent RDF graph, hence, it is not semantics preserving.

We analyze why our direct mapping is not semantics preserving and
realize that monotonicity is an obstacle. We first show that if we
only consider primary keys, we can still have a monotone direct
mapping that is semantics preserving. However this result is not
sufficient because it dismisses foreign keys. Unfortunately, we prove
that no monotone direct mapping is semantics preserving if foreign
keys are considered, essentially because the only form of
constraint checking in OWL is satisfiability testing.
This result 
has an important implication in real world applications: if you
migrate your relational database to the Semantic Web using a monotone
direct mapping, be prepared to experience consistency when what one
would expect is inconsistency.

Finally, we present a non-monotone direct mapping that overcomes the
aforementioned limitation,
We foresee the existence of
monotone direct mappings if OWL is extended with the epistemic
operator. 


\section{Preliminaries}
\noindent
In this section, we define the basic terminology 
used in the paper.

\subsection{Relational databases}
\label{sec-ra}
\noindent
Assume, a countably infinite domain $\D$ and a reserved symbol
$\nil$ that is not in $\D$. A {\em schema} $\R$ is a finite set of
relation names, where for each $R \in \R$, $\att(R)$ denotes the
nonempty finite set of attributes names associated to $R$. An instance
$I$ of $\R$ assigns to each relation symbol $R \in \R$ a finite set
$R^I = \{t_1, \ldots, t_\ell\}$ of tuples, where each tuple $t_j$ ($1
\leq j \leq \ell$) is a function that assigns to each attribute in
$\att(R)$ a value from $(\D \cup \{\nil\})$. We use notation $t.A$ to
refer to the value of a tuple $t$ in an attribute $A$.


\medskip 

\noindent
{\bf Relational algebra:} To define some of the concept studied in
this paper, we use relational algebra as a query language for
relational databases. Given that we consider relational databases
containing null values, we present in detail the syntax and
semantics of a version of relational algebra that formalizes the way
nulls are treated in practice in database systems. Formally, assume
that $\R$ is a relational schema. Then a relational algebra expression
$\varphi$ over $\R$ and its set of attributes $\att(\varphi)$ are
recursively defined as follows:
\smallskip
\begin{alsl}
\item[1.] If $\varphi = R$ with $R \in \R$, then $\varphi$ is a
relational algebra expression over $\R$ such that $\att(\varphi) =
\att(R)$. 

\item[2.] If $\varphi = \nil_A$, where $A$ is an attribute, then
$\varphi$ is a relational algebra expression over $\R$ such that
$\att(\varphi) = \{A\}$.   

\item[3.] If $\psi$ is a relational algebra expression over $\R$,
$A \in \att(\psi)$, $a \in \D$ and $\varphi$ is any of the expressions
$\sigma_{A = a}(\psi)$, $\sigma_{A \neq a}(\psi)$,
$\sigma_{\isn(A)}(\psi)$ or $\sigma_{\isnn(A)}(\psi)$, then $\varphi$
is a relational algebra expression over $\R$ such that $\att(\varphi)
= \att(\psi)$. 

\item[4.] If $\psi$ is a relational algebra expression over $\R$,
$U \subseteq \att(\psi)$ and $\varphi = \pi_U(\psi)$, then $\varphi$
is a relational algebra expression over $\R$ such that $\att(\varphi)
= U$.

\item[5.] If $\psi$ is a relational algebra expression over $\R$,
$A \in \att(\psi)$, $B$ is an attribute such that $B \not\in
\att(\psi)$ and $\varphi = \delta_{A \to B}(\psi)$, then $\varphi$ is a
relational algebra expression over $\R$ such that $\att(\varphi) =
(\att(\psi) \smallsetminus \{A\}) \cup \{B\}$.

\item[6.] If $\psi_1$, $\psi_2$ are relational algebra expressions
over $\R$ and $\varphi = (\psi_1 \nj \psi_2)$, then $\varphi$ is a
relational algebra expression over $\R$ such that $\att(\varphi) =
(\att(\psi_1) \cup \att(\psi_2))$.

\item[7.] If $\psi_1$, $\psi_2$ are relational algebra expressions
over $\R$ such that $\att(\psi_1) = \att(\psi_2)$ and $\varphi$ is
either $(\psi_1 \cup \psi_2)$ or $(\psi_1 \smallsetminus \psi_2)$,
then $\varphi$ is a relational algebra expression over $\R$ such that
$\att(\varphi) =
\att(\psi_1)$. 
\end{alsl}
\smallskip
Let $\R$ be a relational schema, $I$ an instance of $\R$ and $\varphi$
a relational algebra expression over $\R$. The evaluation of $\varphi$
over $I$, denoted by $\semr{\varphi}$, is defined recursively as follows:
\smallskip
\begin{alsl}
\item[1.] If $\varphi = R$ with $R \in \R$, then $\semr{\varphi} =
R^I$. 

\item[2.] If $\varphi = \nil_A$, where $A$ is an attribute, then
$\semr{\varphi} = \{t\}$, where $t : \{A\} \to (\D \cup \{\nil\})$ is a
tuple such that $t.A = \nil$.

\item[3.] Let $\psi$ be a relational algebra expression over $\R$, $A
\in \att(\psi)$ and $a \in \D$. If $\varphi = \sigma_{A = a}(\psi)$,
then $\semr{\varphi} = \{ t \in \semr{\psi} \mid t.A = a\}$.  If
$\varphi = \sigma_{A \neq a}(\psi)$, then $\semr{\varphi} = \{ t \in
\semr{\psi} \mid t.A \neq \nil \text{ and } t.A \neq a\}$.  If
$\varphi = \sigma_{\isn(A)}(\psi)$, then $\semr{\varphi} = \{ t \in
\semr{\psi} \mid t.A = \nil\}$.  If $\varphi =
\sigma_{\isnn(A)}(\psi)$, then $\semr{\varphi} = \{ t \in
\semr{\psi} \mid t.A \neq \nil\}$.





\item[4.] If $\psi$ is a relational algebra expression over $\R$,
$U \subseteq \att(\psi)$ and $\varphi = \pi_U(\psi)$, then
$\semr{\varphi} = \{ t : U \to (\D \cup \{\nil\}) \mid$
there exists $t^\prime \in \semr{\psi}$ such that for every $A \in
U$: $t.A = t^\prime.A\}$.

\item[5.] If $\psi$ is a relational algebra expression over $\R$,
$A \in \att(\psi)$, $B$ is an attribute such that $B \not\in
\att(\psi)$ and $\varphi = \delta_{A \to B}(\psi)$, then
$\semr{\varphi} = \{ t : \att(\varphi) \to (\D \cup \{\nil\}) \mid$
there exists $t^\prime \in \semr{\psi}$ such that $t.B = t^\prime.A$
and for every $C \in (\att(\varphi) \smallsetminus \{B\})$: $t.C =
t^\prime.C\}$.

\item[6.] If $\psi_1$, $\psi_2$ are relational algebra expressions
over $\R$ and $\varphi = (\psi_1 \nj \psi_2)$, then $\semr{\varphi} =
\{ t : \att(\varphi) \to (\D \cup \{\nil\}) \mid$ there exist $t_1 \in
\semr{\psi_1}$ and $t_2 \in \semr{\psi_2}$ such that for every $A \in
(\att(\psi_1) \cap \att(\psi_2))$: $t.A = t_1.A = t_2.A \neq \nil$,
for every $A \in (\att(\psi_1) \smallsetminus \att(\psi_2))$: $t.A =
t_1.A$, and for every $A \in (\att(\psi_2) \smallsetminus \att(\psi_1))$:
$t.A = t_2.A\}$. 

\item[7.] Let $\psi_1$, $\psi_2$ be relational algebra expressions
over $\R$ such that $\att(\psi_1) = \att(\psi_2)$.  If $\varphi =
(\psi_1 \cup \psi_2)$, then $\semr{\varphi} = \semr{\psi_1} \cup
\semr{\psi_2}$.  If $\varphi = (\psi_1 \smallsetminus \psi_2)$, then
$\semr{\varphi} = \semr{\psi_1} \smallsetminus \semr{\psi_2}$.


\end{alsl}
\smallskip
It is important to notice that the operators left-outer join,
right-outer join and full-outer join are all expressible with the
previous operators. For more details, we refer the reader to 
the Appendix.

\medskip 

\noindent
{\bf Integrity constraints:}
We consider two types of integrity constraints: keys and foreign
keys. Let $\R$ be a relational schema. A key $\varphi$ over $\R$ is an
expression of the form $R[A_1, \ldots, A_m]$, where $R \in
\R$ and $\emptyset \subsetneq \{A_1, \ldots, A_m\} \subseteq
\att(R)$. Given an instance $I$ of $\R$, $I$ satisfies key $\varphi$,
denoted by $I \models \varphi$, if: (1) for every $t  \in R^I$ and $k
\in \{1, \ldots, m\}$, it holds that $t.A_k \neq  \nil$, and (2) for
every $t_1, t_2 \in R^I$, if $t_1.A_k = t_2.A_k$ for every $k \in \{1,
\ldots, m\}$, then $t_1 = t_2$. A foreign key over $\R$ is an
expression of the form $R[A_1, \ldots, A_m] \fk S[B_1, \ldots, B_m]$,
where $R, S \in \R$, $\emptyset \subsetneq \{A_1, \ldots, A_m\}
\subseteq \att(R)$ and $\emptyset \subsetneq \{B_1$, $\ldots, B_m\}
\subseteq \att(S)$. Given an instance $I$ of $\R$, $I$ satisfies
foreign key $\varphi$, denoted by $I \models \varphi$, if $I \models
S[B_1, \ldots, B_m]$ and for every tuple $t$ in $R^I$: either (1)
there exists $k \in \{1, \ldots, m\}$ such that $t.A_k = \nil$, or (2)
there exists a tuple $s$ in $S^I$ such that $t.A_k = s.B_k$ for every
$k \in \{1, \ldots, m\}$.

Given a relational schema $\R$, a set $\Sigma$ of keys and foreign
keys is said to be a {\em set of primary keys (PKs) and foreign keys
(FKs) over $\R$} if: (1) for every $\varphi \in \Sigma$, it holds that
$\varphi$ is either a key or a foreign key over $\R$, and (2) there
are no two distinct keys in $\Sigma$ of the form $R[A_1, \ldots, A_m]$
and $R[B_1, \ldots, B_n]$ (that is, that mention the same relation
name $R$). Moreover, an instance $I$ of $\R$ satisfies $\Sigma$,
denoted by $I \models \Sigma$, if for every $\varphi \in \Sigma$, it
holds that $I \models \varphi$.

\subsection{RDF and OWL}
\noindent
Assume there are pairwise disjoint infinite sets $\I$ (IRIs), $\B$
(blank nodes) and $\L$ (literals). A tuple $(s,p,o)\in (\I \cup \B)
\times \I \times (\I \cup \B \cup \L)$ is called an RDF triple, where
$s$ is the subject, $p$ is the predicate and $o$ is the object. A
finite set of RDF triples is called an RDF graph. Moreover, assume the
existence of an infinite set $\V$ of variables disjoint from the above
sets, and assume that every element in $\V$ starts with the symbol
$?$.

In this paper, we consider RDF graphs with OWL
vocabulary~\cite{OWL2}, which is the W3C standard ontology language
based on description logics, without datatypes. 
In particular, we say that an RDF graph $G$ is consistent under OWL 
semantics if a model of $G$ with respect to the OWL vocabulary
exists (see \cite{OWL2} for a precise definition of the notion of
model and the semantics of OWL).

\subsection{SPARQL}
\label{sec-sparql}
\begin{sloppypar}
\noindent
In this paper, we use SPARQL as a query language for RDF graphs. The
official syntax of SPARQL~\cite{PS08,HS11} considers operators
\verb|OPTIONAL|, \verb|UNION|, \verb|FILTER|, \verb|SELECT|, \verb|AS| and
concatenation via a point symbol (\verb|.|), to construct graph
pattern expressions. The syntax of the language also considers
\verb|{ }| to group patterns, and some implicit rules of precedence and 
association.  In order to avoid ambiguities in the parsing, we follow
the approach proposed in \cite{PAG09}, and we present the syntax of
SPARQL graph patterns in a more traditional algebraic formalism, using
operators $\AND$ (\verb|.|), $\UNION$ (\verb|UNION|), $\OPT$
(\verb|OPTIONAL|), $\MINUS$ (\verb|MINUS|), $\FILTER$ (\verb|FILTER|),
$\SELECT$ (\verb|SELECT|) and $\AS$ (\verb|AS|). More precisely, a
SPARQL graph pattern expression is defined recursively as follows.
\end{sloppypar}
\smallskip
\begin{alsl}
\item[1.] $\egp$ is a graph pattern (the empty graph pattern). 

\item[2.] A tuple from $(\I \cup \L \cup \V) \times (\I \cup \V) \times
(\I \cup \L \cup \V)$ is a graph pattern (a triple pattern).

\begin{sloppypar}
\item[3.] If $P_1$ and $P_2$ are graph patterns, then expressions
$(P_1 \andp P_2)$, $(P_1 \opt P_2)$, $(P_1 \uni P_2)$ and $(P_1 \minus
P_2)$ are graph patterns.
\end{sloppypar}

\item[4.] If $P$ is a graph pattern and $R$ is a SPARQL built-in
condition, then the expression $(P \vc R)$ is a graph pattern.

\item[5.] If $P$ is a graph pattern and $?A_1$, $\ldots$, $?A_m$, $?B_1$,
$\ldots$, $?B_m$, $?C_1$, $\ldots$, $?C_n$ is a sequence of pairwise
distinct elements from $\V$ ($m \geq 0$ and $n \geq 0$) such that none
of the variables $?B_i$ ($1 \leq i \leq m$) is mentioned in $P$, then
{\small
\begin{eqnarray*}
\sel{\{?A_1 \as ?B_1, \ldots , ?A_m \as ?B_m, ?C_1, \ldots,
?C_n\}}{P}
\end{eqnarray*}
\normalsize{is}} a graph pattern.
  
\end{alsl}
\smallskip
A SPARQL built-in condition is constructed using elements of the set
$(\I \cup \V)$ and constants, logical connectives ($\neg$, $\wedge$,
$\vee$), inequality symbols ($<$, $\leq$, $\geq$, $>$), the equality
symbol ($=$), unary predicates such as $\bound$, $\isBlank$, and
$\isIRI$ (see~\cite{PS08,HS11} for a complete
list).  In this paper, we restrict to the fragment where the built-in
condition is a Boolean combination of terms constructed by using $=$
and $\bound$, that is: (1) if $?X, ?Y \in \V$ and $c \in \I$, then
$\bound(?X)$, $?X = c$ and $?X = ?Y$ are built-in conditions, and (2)
if $R_1$ and $R_2$ are built-in conditions, then $(\neg R_1)$, $(R_1
\vee R_2)$ and $(R_1 \wedge R_2)$ are built-in conditions. 

\begin{sloppypar}
The version of SPARQL used in this paper includes the following SPARQL 1.1 features:
the operator $\MINUS$, the possibility of nesting the
$\SELECT$ operator and the operator $\AS$~\cite{HS11}.
\end{sloppypar}

The answer of a SPARQL query $P$ over an RDF graph $G$ is a finite set
of {\em mappings}, where a mapping $\mu$ is a partial
function from the set $\V$ of variables to $(\I \cup \L \cup
\B)$. 
We define the semantics of SPARQL as a function
$\sem{\,\cdot\,}$ that, given an RDF graph $G$, takes a graph pattern
expression and returns a set of mappings. We refer the reader to
the Appendix for more detail.

 
\section{Direct Mappings: Definition and Properties}
\label{sec-fp}
\noindent
A direct mapping is a default way to translate relational
databases into RDF (without any input from the user on how
the relational data should be translated). The input of a direct
mapping $\M$ is a relational schema $\R$, a set $\Sigma$ of PKs and
FKs over $\R$ and an instance $I$ of $\R$. The output is an RDF graph
with OWL vocabulary.

Assume $\rdf$ is the set of all 
RDF graphs
and $\rel$ is the set of all triples of the form $(\R, \Sigma, I)$
such that $\R$ is a relational schema, $\Sigma$ is a set of PKs and
FKs over $\R$ and $I$ is an instance of $\R$.
\begin{definition}[Direct mapping]
A direct mapping $\M$ is a total function from $\rel$ to $\rdf$. 
\end{definition}
We now introduce two fundamental properties
of direct mappings: information preservation and query preservation;
and two desirable properties of these mappings: monotonicity and
semantic preservation. Information preservation is a fundamental
property because it guarantees that the mapping does not lose
information, which is fundamental in an Extract-Transform-Load
process. Query preservation is also a fundamental property because it
guarantees that 
everything that can be extracted from the relational
database by 
a relational algebra query, can also be extracted
from the resulting RDF graph by 
a SPARQL query. This property
is fundamental for workloads that involve translating SPARQL
to SQL. Monotonicity is a desirable property because it would
avoid recalculating the mapping for the entire database after
inserting new data. In addition to practical considerations when translating relational data to RDF graphs, we must deal with the closed-world database semantics and open world RDF/OWL semantics. Understanding the expressive power of a mapping and, its ability to properly deal with integrity constraints is important. Thus our choice of examining semantics preservation.

\subsection{Fundamental properties}


\noindent
{\bf Information preservation:}
A direct mapping is information preserving if it does not
lose any information about the relational instance being translated,
that is, if there exists a way to recover the original database instance
from the RDF graph resulting from the translation process. Formally,
assuming that $\ins$ is the set of all possible relational instances,
we have that:
\begin{sloppypar}
\begin{definition}[Information preservation]\label{def-ip}
A direct mapping $\M$ is {\em information preserving} if there
is a computable mapping $\N : \rdf \to \ins$ such that for every
relational schema $\R$, set $\Sigma$ of PKs and FKs over $\R$, and
instance $I$ of $\R$ satisfying $\Sigma$:
$\N(\M(\R, \Sigma, I)) = I$.
\end{definition}
\end{sloppypar}
\noindent
Recall that a mapping $\N : \rdf \to \ins$ is computable if there
exists an algorithm that, given $G \in \rdf$, computes $\N(G)$. 

\medskip

\noindent
{\bf Query preservation:}
A direct mapping is query preserving if every query over
a relational database can be translated into an equivalent query over
the RDF graph resulting from the mapping. That is, query preservation
ensures that every relational query can be evaluated using the mapped
RDF data.

To formally define query preservation, we focus on relational queries
that can be expressed in relational algebra~\cite{AHV95} and RDF
queries that can be expressed in SPARQL~\cite{PS08,PAG09}. In Section
\ref{sec-ra}, we introduced a version of relational algebra that
formalizes the semantics of null values in practice. In Section
\ref{sec-sparql}, we introduce an algebraic version of SPARQL that
follows the approach proposed in \cite{PAG09}. Given the mismatch in
the formats of these query languages, we introduce a function $\tr$
that converts tuples returned by relational algebra
queries into mappings returned by SPARQL. Formally, given a relational
schema $\R$, a relation name $R \in \R$, an instance $I$ of $\R$ and a
tuple $t \in R^I$, define $\tr(t)$ as the mapping $\mu$ such
that: (1) the domain of $\mu$ is $\{ ?A \mid A \in \att(R)$ and $t.A
\neq \nil\}$, and (2) $\mu(?A) = t.A$ for every $A$ in the
domain~of~$\mu$. 

\begin{example}{\em
Assume that a relational schema contains a relation name {\tt
STUDENT} and attributes {\tt ID}, {\tt NAME} and {\tt AGE}. Moreover,
assume that $t$ is a tuple in this relation such that $t.\text{\tt ID}
= \text{1}$, $t.\text{\tt NAME} = \text{John}$ and $t.\text{\tt AGE} =
\nil$. Then, $\tr(t) = \mu$, where the domain of $\mu$ is
$\{\text{\tt ?ID},\, \text{\tt ?NAME}\}$, $\mu(\text{\tt ?ID}) =
\text{1}$ and $\mu(\text{\tt ?NAME}) = \text{John}$. \qed}
\end{example}


\begin{sloppypar}
\begin{definition}[Query preservation]
A direct mapping $\M$ is {\em query preserving} if for every relational
schema $\R$, set $\Sigma$ of PKs and FKs over $\R$ and relational
algebra query $Q$ over $\R$, there exists a SPARQL query $Q^\star$
such that for every instance $I$ of $\R$ satisfying $\Sigma$:
$\tr(\semr{Q}) = \semp{Q^\star}{\M(\R,\Sigma,I)}$.
\end{definition}
\end{sloppypar}
\noindent
It is important to notice that information preservation and query
preservation are incomparable properties in our setting. On one side,
if a direct mapping $\M$ is information preserving, this does not
guarantee that every relational algebra query $Q$ can be rewritten
into an equivalent SPARQL query over the translated data, as $\M$
could transform source relational databases in such a way that a more
expressive query language is needed to express $Q$ over the generated
RDF graphs. On the other side, a mapping $\M$ can be query preserving
and not information preserving if the information about the schema of
the relational database being translated is not stored. For example,
we define in Section \ref{sec-dm} a direct mapping $\DM$ that includes
information about these relational schemas.  It will become clear in
Sections \ref{sec-dm} and \ref{sec-lossless} that if such information
is not stored, then $\DM$ would be query preserving but not information
preserving.

\subsection{Desirable properties}
\label{sec-desi}
\noindent
{\bf Monotonicity:}
Given two database instances $I_1$ and $I_2$ over a relational schema
$\R$, instance $I_1$ is said to be contained in instance $I_2$,
denoted by $I_1
\subseteq I_2$, if for every $R \in \R$, it holds that $R^{I_1}
\subseteq R^{I_2}$. A direct mapping $\M$ is considered monotone if
for any such pair of instances, the result of mapping $I_2$ contains
the result of mapping $I_1$. In other words, if we insert new data to
the database, then the elements of the mapping that are already
computed are unaltered. 

\begin{sloppypar}
\begin{definition}
[Monotonicity]\label{def-mono-dm}
A direct mapping $\M$ is monotone if for every relational schema $\R$,
set $\Sigma$ of PKs and FKs over $\R$, and instances $I_1$, $I_2$
of $\R$ such that $I_1 \subseteq I_2$: 
$\M(\R, \Sigma, I_1) \subseteq \M(\R, \Sigma, I_2)$.
\end{definition}
\end{sloppypar}

\medskip

\noindent
{\bf Semantics preservation:}
A direct mapping is semantics preserving if the
satisfaction of a set of PKs and FKs by a relational database is
encoded in the translation process. More precisely, given a relational
schema $\R$, a set $\Sigma$ of PKs and FKs over $\R$ and an instance
$I$ of $\R$, a semantics preserving mapping should generate from $I$ a
consistent RDF graph if $I \models \Sigma$, and it should generate an
inconsistent RDF graph otherwise.
\begin{definition}[Semantics preservation]
A direct mapping $\M$ is semantics preserving if for every relation
schema $\R$, set $\Sigma$ of PKs and FKs over $\R$ and instance $I$ of
$\R$:
$I \models \Sigma$ \ \ iff \ \ $\M(\R, \Sigma, I)$ is
consistent under OWL semantics.
\end{definition}


\section{The Direct Mapping $\DM$}
\label{sec-dm}
\noindent
We introduce a direct mapping $\DM$, that integrates
and extends the functionalities of the direct mappings proposed
in~\cite{TSM08,APS11}. $\DM$ is defined as a set of Datalog
rules\footnote{We refer the reader to \cite{AHV95} for the syntax and
semantics of Datalog.}, which are divided in two parts: translate relational schemas and
translate relational instances.


In Section \ref{sec-sro}, we present the predicates that
are used to store a relational database, the input of $\DM$. In Section
\ref{sec-ro}, we present predicates that are used
to store an ontology and Datalog rules to generate an ontology from the relational schema and
the set of PKs and FKs. In Section \ref{sec-dms}, we
present the Datalog rules that generate the OWL vocabulary from the ontology that was derived from the relational schema and a set of
PKs and FKs. Finally, we present in Section \ref{sec-dmi} the Datalog
rules that  generates RDF triples from a relational instance.

\begin{sloppypar}
Throughout this section, we use the following running example. Consider a
relational database for a university. The schema of this database
consists of tables {\tt STUDENT(SID,NAME)}, {\tt
COURSE(CID,TITLE,CODE)}, {\tt DEPT(DID,NAME)} and {\tt
ENROLLED(SID,CID)}. 
Moreover, we have the following
constraints about the schema of the university: {\tt SID} is the
primary key of {\tt STUDENT}, {\tt CID} is the primary key of {\tt
COURSE}, {\tt DID} is the primary key of {\tt DEPT}, {\tt (SID,CID)}
is the primary key of {\tt ENROLLED}, {\tt CODE} is a foreign key in
{\tt COURSE} that references attribute {\tt DID} in {\tt DEPT}, {\tt
SID} is a foreign key in {\tt ENROLLED} that references attribute {\tt
SID} in {\tt STUDENT}, and {\tt CID} is a foreign key in {\tt
ENROLLED} that references attribute {\tt CID} in {\tt COURSE}.
\end{sloppypar}

\subsection{Storing relational databases}
\label{sec-sro}
\noindent
Given that the direct mapping $\DM$ is specified by a set
of Datalog rules, its input $(\R,\Sigma,I)$ has to be encoded as a set
of relations. We define the predicates that are used
to store the triples of the form $(\R,\Sigma,I)$. More precisely, the
following predicates are used to store a relational schema $\R$ and a
set $\Sigma$ of PKs and FKs over $\R$.
\smallskip
\begin{alsl}
\begin{sloppypar}
\item[$\bullet$] $\Rel(r)$: Indicates that $r$ is a relation name in $\R$;
e.g. $\Rel(\STUDENT)$ indicates that {\tt STUDENT} is a relation
name.\footnote{As is customary, we use double quotes to delimit
strings.}

\item[$\bullet$] \begin{sloppypar}$\Attr(a, r)$: Indicates that $a$ is an
attribute in the relation $r$ in $\R$; e.g. $\Attr(\NAME$, $\STUDENT)$ holds. 
\end{sloppypar}

\item[$\bullet$] \begin{sloppypar} $\PK_n(a_1, \ldots, a_n, r)$: Indicates that
$r[a_1, \ldots, a_n]$ is a primary key in $\Sigma$;
e.g. {\small $\PK_1(\SID,\STUDENT)$} holds. 
\end{sloppypar}

\item[$\bullet$] $\FK_n(a_1, \ldots, a_n, r, b_1, \ldots, b_n, s)$: Indicates
that $r[a_1, \ldots, a_n] \fk s[b_1, \ldots, b_n]$ is a foreign key in
$\Sigma$; e.g. $\FK_1(\CODE,\COURSE,\DID,\DEPT)$ holds.
\end{sloppypar}
\end{alsl}
\smallskip
Moreover, the following predicate is used to store the tuples in an
relational instance $I$ of a relational schema $\R$.
\smallskip
\begin{alsl}
\item[$\bullet$] \begin{sloppypar} $\Value(v, a, t, r)$: Indicates that $v$ is
the value of an attribute $a$ in a tuple with identifier $t$ in a
relation $r$ (that belongs to $\R$); e.g. a tuple $t_1$ of table {\tt
STUDENT} such that $t_1.\text{\tt SID} = \text{\tt "1"}$ and
$t_1.\text{\tt NAME} = \nil$ is stored by using the facts $\Value(\text{\tt
"1"}, \SID, \text{\tt "id1"}, \STUDENT)$ and $\Value(\nil, \NAME, \text{\tt "id1"}$, $\STUDENT)$, assuming that
\text{\tt id1} is the identifier of tuple $t_1$. \end{sloppypar}
\end{alsl}

\subsection{Storing an ontology}
\label{sec-ro}
\noindent
In order to translate a relational database into an RDF graph
with OWL vocabulary, we first extract an ontology from the
relational schema and the set of PKs and FKs given as
input. In particular, we classify each relation name in the schema as a
class or a binary relation (which is used to represent a many-to-many
relationship between entities in an ER/UML diagram), we represent
foreign keys as object properties and attributes of
relations as data type properties. More specifically, the following
predicates are used to store the extracted ontology:
\smallskip
\begin{alsl}
\item[$\bullet$] $\Class(c)$: Indicates that $c$ is a class.
 
\item[$\bullet$] $\ObjP_n(p_1, \ldots, p_n, d, r)$: Indicates that
$p_1, \ldots, p_n$ ($n \geq 1$) form an object property with domain
$d$ and range $r$.

\item[$\bullet$] $\DTP(p,d)$: Indicates that $p$ is a data type property with domain $d$.
\end{alsl}
\smallskip
The above predicates are defined by the Datalog rules described in the
following sections.

\medskip
\noindent
{\bf Identifying binary relations:}  
We define auxiliary predicates that identify binary relations to
facilitate identifying classes, object properties and data type
properties. Informally, a relation $R$ is a binary relation between
two relations $S$ and $T$ if (1) both $S$ and $T$ are different from
$R$, (2) $R$ has exactly two attributes $A$ and $B$, which form a
primary key of $R$, (3) $A$ is the attribute of a foreign key in $R$
that points to $S$, (4) $B$ is the attribute of a foreign key in $R$
that points to $T$, (5) $A$ is not the attribute of two distinct
foreign keys in $R$, (6) $B$ is not the attribute of two distinct
foreign keys in $R$, (7) $A$ and $B$ are not the attributes of a
composite foreign key in $R$, and (8) relation $R$ does not have
incoming foreign keys. In Datalog this becomes:

\vspace{-10pt}

{\small
\begin{align}
\notag
\BinRel(&R, A, B, S, C, T, D) \leftarrow \\ \notag 
        & \PK_2(A, B, R), \neg \ExistThreeAttr(R),\\  \notag 
	&\FK_1(A, R, C, S), R \neq S, \FK_1(B, R, D, T), R \neq T,\\ 
	&\neg \TwoFKfrom(A, R), \neg \TwoFKfrom(B, R), \label{eq-binrel} \\ 
	& \neg \FKfrom(A, B, R), \neg \FKto(R). \notag
\end{align}
\bw{In}} a Datalog rule, negation is represented with the symbol
$\neg$ and upper case letters are used to denote variables. Thus, the
previous rule states that the relation $R$ is a binary relation
between two relations $S$ and $T$ if the following conditions are
satisfied. (a) Expression $\PK_2(A, B, R)$ in \eqref{eq-binrel}
indicates that  attributes $A$ and $B$ form a primary key of $R$. (b)
Predicate $\ExistThreeAttr$ checks whether a relation has at least
three attributes, and it is defined as follows from the base predicate
$\Attr$: 
%
%
{\small
\begin{eqnarray*}
\ExistThreeAttr(R) & \leftarrow & \Attr(X, R), \Attr(Y, R),\\ 
&&\Attr(Z, R), X \neq Y, X \neq Z, Y \neq Z.
\end{eqnarray*}
\bw{Thus,}} expression $\neg \ExistThreeAttr(R)$ in \eqref{eq-binrel}
indicates that $R$ has at least two attributes. Notice that by
combining this expression with $\PK_2(A, B, R)$, we conclude that $A$,
$B$ are exactly the attributes of $R$. (c) Expressions $\FK_1(A, R, C,
S)$ and $\FK_1(B, R, D, T)$ in \eqref{eq-binrel} indicate that $A$ is
the attribute of a foreign key in $R$ that points to $S$ and $B$ is
the attribute of a foreign key in $R$ that points to $T$,
respectively. (d) Expressions $R \neq S$ and $R \neq T$ in
\eqref{eq-binrel} indicate that both $S$ and $T$ are different from
relation $R$. (e) Predicate $\TwoFKfrom$ checks whether an attribute of
a relation is the attribute of two distinct foreign keys in that
relation, and it is defined as follows from the base predicate
$\FK_1$: 

\vspace{-10pt}

{\small
\begin{eqnarray*}
\TwoFKfrom(X,Y) &\leftarrow& \FK_1(X, Y, U_1, V_1), \FK_1(X, Y, U_2, V_2), \\
			    &                & U_1 \neq U_2 \\	 
\TwoFKfrom(X,Y) &\leftarrow& \FK_1(X, Y, U_1, V_1), \FK_1(X, Y, U_2, V_2), \\
			    &                 & V_1 \neq V_2
\end{eqnarray*}
\bw{Thus,}} expressions $\neg \TwoFKfrom(A, R)$ and $\neg \TwoFKfrom(B, R)$
in \eqref{eq-binrel} indicate that attribute $A$ is not the attribute
of two distinct foreign keys in $R$ and $B$ is not the attribute of
two distinct foreign keys in $R$, respectively. (f) Predicate
$\FKfrom$ checks whether a pair of attributes of a  relation are the
attributes of a composite foreign key in that relation:

\vspace{-10pt}

{\small
\begin{eqnarray*}
\FKfrom(X,Y,Z) & \leftarrow & \FK_2(X,Y,Z,U,V,W) \\
\FKfrom(X,Y,Z) & \leftarrow & \FK_2(Y,X,Z,U,V,W)
\end{eqnarray*}
\bw{Thus,}} expression $\neg \FKfrom(A, B, R)$ in \eqref{eq-binrel}
indicates that attributes $A$, $B$ of $R$ are not the attributes of a
composite foreign key in $R$. (g) Finally, predicate $\FKto$ checks
whether a relation with two attributes has incoming foreign keys: 

\vspace{-10pt}

{\small
\begin{eqnarray*}
\FKto(X) & \leftarrow & \FK_1(U_1,Y,V,X)\\
\FKto(X) & \leftarrow & \FK_2(U_1,U_2,Y,V_1,V_2,X)
\end{eqnarray*}
\bw{Thus,}} expression $\neg \FKto(R)$ in \eqref{eq-binrel} indicates that
relation $R$ does not have incoming foreign keys.

\begin{sloppypar}
For instance, $\BinRel(\ENROLLED$, $\SID$, $\CID$, $\STUDENT$, $\SID$,
$\COURSE$, $\CID)$ holds in our example. Note that there is no
condition in the rule
\eqref{eq-binrel} that requires $S$ and $T$ to be different, allowing
binary relations that have their domain equal to their range. Also
note that, for simplicity, we assume in the rule \eqref{eq-binrel}
that a binary relation $R$ consists of only two attributes $A$ and
$B$. However, this rule can be easily extended to deal with binary
relations generated from many-to-many relationships between entities
in an ER/UML diagram that have more than two attributes.
\end{sloppypar}

\medskip

\noindent
{\bf Identifying classes:}  
In our context, a class is any relation that is not a binary
relation. That is, predicate $\Class$ is defined by the following
Datalog rules: 

\vspace{-10pt}

{\small
\begin{eqnarray*}
\Class(X) & \leftarrow & \Rel(X), \neg \IsBinRel(X)\\	
\IsBinRel(X) & \leftarrow & \BinRel(X, A, B, S, C, T, D)
\end{eqnarray*}}

\vspace{-10pt}

\begin{sloppypar}
\noindent
In our example, $\Class(\DEPT)$, $\Class(\STUDENT)$ and
$\Class(\COURSE)$ hold. 
\end{sloppypar}

 



\medskip

\noindent
{\bf Identifying object properties:} 
\noindent
For every $n \geq 1$, the following rule is used for identifying
object properties that are generated from foreign keys:
\footnote{Notice that although we consider an infinite number of rules
in the definition of $\DM$, for every concrete relational database we
will need only a finite number of these rules.}  {\small
\begin{multline*}
\ObjP_{2n}(X_1, \ldots, X_n, Y_1, \ldots, Y_n, S, T) \ \leftarrow\\
\FK_n(X_1, \ldots, X_n, S, Y_1, \ldots, Y_n, T), \neg \IsBinRel(S)
\end{multline*}
\bw{This}} rule states that a foreign key represents an object property from
the entity containing the foreign key (domain) to the referenced
entity (range). It should be noticed that this rule excludes the case
of binary relations, as there is a special rule for this type of
relations (see rule \eqref{eq-binrel}). In our example,
$\ObjP_2(\CODE$, $\DID$, $\COURSE$, $\DEPT)$ holds as {\tt CODE} is a
foreign key in the table {\tt COURSE} that references attribute {\tt
DID} in the table {\tt DEPT}.

\medskip

\noindent
{\bf Identifying data type properties:}  
Every attribute in a
non-binary relation is mapped to a data type property:
{\small
\begin{eqnarray*}
\DTP(A, R) & \leftarrow & \Attr(A,R), \neg \IsBinRel(R)
\end{eqnarray*}
\bw{For}} instance, we have that $\DTP(\NAME$,
$\STUDENT)$ holds in our example, while $\DTP(\SID,
\ENROLLED$) does not hold as {\tt ENROLLED} is a binary relation.


\subsection{Translating a relational schema into OWL}
\label{sec-dms}
\noindent
We now define the rules that translates a relational
database schema into an OWL vocabulary. 

\subsubsection{Generating IRIs for classes, object properties and data
type properties}  
\label{sec-iri}
\begin{sloppypar} 
\noindent
We introduce a family
of rules that produce IRIs for classes, binary relations,
object properties and data type properties identified by the mapping
(which are stored in the predicates $\Class$, $\BinRel$, $\ObjP_n$ and
$\DTP$, respectively). Note that the IRIs generated can be later on replaced or mapped to existing IRIs available in the Semantic Web. Assume given a base IRI $\base$
for the relational database to be translated (for example, {\tt
"http://example.edu/db/"}), and assume given a family of built-in
predicates $\concat_n$ ($n \geq 2$) such that $\concat_n$ has $n+1$
arguments and $\concat_n(x_1, \ldots, x_n, y)$ holds if $y$ is the
concatenation of the strings $x_1$, $\ldots$, $x_n$. Then by following
the approach proposed in
\cite{APS11}, $\DM$ uses the following Datalog rules to produce IRIs
for classes and data type properties:
\end{sloppypar}

\vspace*{-10pt}

{\small
\begin{align*}
\gClassIRI(R, X) & \leftarrow \Class(R), \concat_2(\base, R, X)\\
\gDTPIRI(A, R, X) & \leftarrow \DTP(A,R), \concat_4(\base, R,
\text{\tt "\#"}, A, X) 
\end{align*}}

\vspace*{-10pt}

\begin{sloppypar}
\noindent
For instance, {\tt http://example.edu/db/STUDENT} is the IRI
for the {\tt STUDENT} relation in our example, and 
{\tt http://example.edu/db/STUDENT\#NAME} 
is the IRI for attribute {\tt
NAME} in the {\tt STUDENT} relation (recall that $\DTP(\NAME$,
$\STUDENT)$ holds in our example). Moreover, $\DM$ uses the
following family of Datalog rules to generate IRIs for object
properties. First, for object properties generated from binary
relations, the following rules is used:
\end{sloppypar}

\vspace*{-10pt}

{\small
\begin{align*}
\gObjPIRI_1(&R, A, B, S, C, T, D, X) \leftarrow\\
&\BinRel(R, A, B, S, C, T, D), \\
&\concat_{10}(\base, R, \text{\tt "\#"}, A, \text{\tt ","}, B, \text{\tt ","}, C, \text{\tt ","}, D,  X)
\end{align*}}

\vspace*{-10pt}

\noindent
Thus, {\small {\tt http://example.edu/db/ENROLLED\#SID,CID,SID,CID}}
is the IRI for binary relation {\tt ENROLLED} in our
example. Second, for object properties generated from a foreign key
consisting of $n$ attributes ($n \geq 1$), the following rule is used:

\vspace{-10pt}

{\small
\begin{align*}
\text{\sc OP\_}&\text{\sc IRI}_{2n}(X_1, \ldots, X_n, Y_1, \ldots,
Y_n, S, T, X) \leftarrow \\ &\ObjP_{2n}(X_1, \ldots, X_n, Y_1, \ldots,
Y_n, S, T),\\   
&\concat_{4n+4}(\base, S, \text{\tt ","}, T, \text{\tt "\#"}, X_1,
\text{\tt ","}, \ldots, X_{n-1}, \text{\tt ","},  \\  
&\hspace{50pt}X_n, \text{\tt ","}, Y_1, \text{\tt ","}, \ldots,
Y_{n-1}, \text{\tt ","}, Y_n,X)  
\end{align*}}

\vspace{-10pt}

\begin{sloppypar}
\noindent
Thus, given that $\ObjP_2(\CODE$, $\DID$, $\COURSE$, $\DEPT)$ holds in
our example, IRI {\tt
http://example.edu/db/COURSE,DEPT\#CODE,DID} is generated to
represented the fact that {\tt CODE} is a foreign key in the table
{\tt COURSE} that references attribute {\tt DID} in the table {\tt
DEPT}. 
\end{sloppypar}

\subsubsection{Translating relational schemas} 
\noindent
The following Datalog rules are used to generate the RDF
representation of the OWL vocabulary. First, a rule is used to collect
all the classes:

\vspace{-10pt}

{\small
\begin{multline*}
\Triple(U, \rdftype, \owlClass)	\ \leftarrow  \\ \Class(R), \gClassIRI(R, U)
\end{multline*}
\bw{Predicate}} $\Triple$ is used to collect all the
triples of the RDF graph generated by the direct mapping $\DM$. Second,
the following family of rules is used to collect all the object
properties ($n \geq 1$):

\vspace{-10pt}

{\small
\begin{multline*}
\Triple(U, \rdftype, \owlOP) \ \leftarrow\\ 
\ObjP_n(X_1, \ldots, X_n, S, T), \gObjPIRI_n(X_1, \ldots, X_n, S, T, U)
\end{multline*}
\bw{Third,}} the following rule is used to collect the domains of
the object properties ($n \geq 1$):

\vspace{-10pt}

{\small
\begin{multline*}
\Triple(U, \domain, W) \leftarrow \ObjP_n(X_1, \ldots, X_n, S, T),\\
\gObjPIRI_n(X_1, \ldots, X_n, S, T, U), \gClassIRI(S, W) 
\end{multline*}
\bw{Fourth,}} the following rule is used to collect the ranges
of the object properties ($n \geq 1$):

\vspace{-10pt}

{\small
\begin{multline*}
\Triple(U, \range, W) \leftarrow \ObjP_n(X_1, \ldots, X_n, S, T),\\ 
\gObjPIRI_n(X_1, \ldots, X_n, S, T, U), \gClassIRI(T, W)
\end{multline*}
\bw{Fifth,}} the following rule is used to collect all the data type
properties: 

\vspace{-10pt}

{\small
\begin{multline*}
\Triple(U, \rdftype, \owlDTP) \ \leftarrow\\ \DTP(A, R),
\gDTPIRI(A, R, U)
\end{multline*}
\bw{Finally,}} the following rule is used to collect the domains of
the data type properties:

\vspace{-10pt}

{\small
\begin{multline*}
\Triple(U, \domain, W) \ \leftarrow\\ \DTP(A, R),  \gDTPIRI(A, R, U),
\gClassIRI(R, W)  
\end{multline*}}


\subsection{Translating a database instance into RDF}
\label{sec-dmi}
\noindent
We now define the rules that map a relational database
instance into RDF. More specifically, we first introduce a series of
rules for generating IRIs, and then we present the Datalog rules that
generate RDF. 

\subsubsection{Generating IRIs for tuples}
\begin{sloppypar}
\noindent
We introduce a family of
predicates that produce IRIs for the tuples being translated, where we
assume a given a base IRI $\base$ for the relational database (for
example, {\tt "http://example.edu/db/"}). First, $\DM$ uses the
following Datalog rule to produce IRIs for the tuples of the relations
having a primary key: 
\end{sloppypar}

\vspace{-10pt}

{\small
\begin{align*}
\gRIRI_n(V_1, V_2, &\ldots, V_n, A_1, A_2, \ldots, A_n, T, R, X) \
\leftarrow \\  
	       &\PK_n(A_1, A_2, \ldots, A_n, R), \Value(V_1, A_1, T, R), \\ 
	       & \Value(V_2,  A_2, T, R), \ldots, \Value(V_n, A_n, T, R),\\ 
	       &\concat_{4n+2}(\base, R, \text{\tt "\#"}, A_1, \text{\tt "="}, V_1, \text{\tt ","}, \\ 
	       & A_2, \text{\tt "="}, V_2, \text{\tt ","}, \ldots, \text{\tt ","}, A_n, \text{\tt "="}, V_n, X)
\end{align*}}

\vspace{-10pt}

\begin{sloppypar}
\noindent
Thus, given that the facts $\PK_1(\SID,\STUDENT)$ and
$\Value(\text{\tt "1"}, \SID, \text{\tt "id1"}, \STUDENT)$ hold in our example, the IRI {\tt
http://example.edu/db/STUDENT\#SID=1} is the identifier for the tuple
in table {\tt STUDENT} with value {\tt 1} in the primary
key. Moreover, $\DM$ uses the following rule to generate blank
nodes for the tuples of the relations not having a primary 
key:\end{sloppypar}

\vspace*{-10pt}

{\small 
\begin{multline*}
\gRBN(T, R, X) \ \leftarrow \\ 
\Value(V, A, T, R), \concat_3(\text{\tt "\_:"}, R, T, X)  
\end{multline*}}
\subsubsection{Translating relational instances}
\noindent
The direct mapping $\DM$ generates three types of triples when
translating a relational instance: Table triples, reference triples
and literal triples~\cite{APS11}. Following are the Datalog rules
for each one of these cases.

For table triples, $\DM$ produces for each tuple $t$ in a relation
$R$, a triple indicating that $t$ is of type $r$. To construct these tuples, $\DM$ uses the following
auxiliary rules:

\vspace{-10pt}

{\small
\begin{align*}
\gID(T, R, & X) \leftarrow\\
& \Class(R), \PK_n(A_1, \ldots, A_n, R),\\
& \Value(V_1, A_1, T, R), \ldots, \Value(V_n, A_n, T, R), \\ 
& \gRIRI_n(V_1, \ldots , V_n, A_1, \ldots, A_n, T, R, X)\\
\gID(T, R, & X) \leftarrow \\
& \Class(R), \neg \HasPK_n(R),\\
& \Value(V, A, T, R), \gRBN(T, R, X)
\end{align*}
\bw{That}} is, $\gID(T, R, X)$ generates the identifier $X$ of a tuple
$T$ of a relation $R$, which is an IRI if $R$ has a primary key or
a blank node otherwise. Notice that in the preceding rules, predicate
$\HasPK_n$ is used to check whether a table $R$ with $n$ attributes
has a primary key (thus, $\neg \HasPK_n(R)$ indicates that $R$ does
not have a primary key). Predicate $\HasPK_n$ is defined by the
following $n$ rules:

\vspace*{-10pt}

{\small
\begin{eqnarray*}
\HasPK_n(X) & \leftarrow & \PK_i(A_1, \ldots, A_i, X) \ \ \ \ \ \ \ \
\ \ \ \ i \in \{1, \ldots, n\}  
\end{eqnarray*}
\bw{The}} following rule generates the table
triples: 
{\small
\begin{multline*}
\Triple(U, \rdftype, W) \ \leftarrow \\ 
\Value(V, A, T, R), \gID(T, R, U), \gClassIRI(R, W)
\end{multline*}
\bw{For}} example, the following is a table triple in our example:
{\small
\begin{align*}
\Triple(&\text{\tt "http://example.edu/db/STUDENT\#SID=1"},\\
&\rdftype,\\ 
&\text{\tt "http://example.edu/db/STUDENT"}) 
\end{align*}
\bw{For}} reference triples, $\DM$ generates triples that
store the references generated by binary relations and foreign keys. More precisely, the following
Datalog rule is used to construct reference triples for object
properties that are generated from binary relations:

\vspace{-10pt}

{\small
\begin{align*}
\Triple(U, V, W) \leftarrow \ & \BinRel(R, A, B, S, C, T, D),\\
					& \Value(V_1, A, T_1, R), \Value(V_1, C, T_2, S), \\
					& \Value(V_2, B, T_1, R), \Value(V_2, D, T_3, T), \\
					& \gID(T_2, S, U),\\
&\gObjPIRI_1(R, A, B, S, C, T, D, V), \\
					&  \gID(T_3, T, W)
\end{align*}
\bw{Moreover,}} the following Datalog rule is used to construct
reference triples for object properties that are generated from
foreign keys ($n \geq 1$):

\vspace{-10pt}

{\small
\begin{align*}
\Triple(U, V, &W) \leftarrow \\ 
&\ObjP_{2n}(A_1, \ldots, A_n, B_1, \ldots, B_n, S, T),\\
					 & \Value(V_1, A_1, T_1, S), \ldots, \Value(V_n, A_n, T_1, S),\\ 
					 & \Value(V_1, B_1, T_2, T), \ldots, \Value(V_n, B_n, T_2, T),\\ 
					 & \gID(T_1, S, U), \gID(T_2, T, W), \\ 
					 & \gObjPIRI_{2n}(A_1, \ldots, A_n, B_1, \ldots, B_n, S, T, V)
\end{align*}
\bw{Finally,}} $\DM$ produces for every tuple $t$ in a
relation $R$ and for every attribute $A$ of $R$, a triple storing the
value of $t$ in $A$, which is called a literal triple.
The following Datalog rule is used to generate such triples:

\vspace{-10pt}

{\small
\begin{multline*}
\Triple(U, V, W) \leftarrow \ \DTP(A, R),  \Value(W, A, T, R),\\ 
W \neq \nil, \gID(T, R, U), \gDTPIRI(A, R, V)
\end{multline*}
\bw{Notice}} that in the above rule, we use the condition $W \neq
\nil$ to check that the value of the attribute $A$ in a tuple $T$ in a
relation $R$ is not null. Thus, literal triples are generated only for
non-null values. The following is an example of a literal triple:

\vspace{-10pt}

{\small
\begin{align*}
\Triple(&\text{\tt "http://example.edu/db/STUDENT\#SID=1"},\\
&\text{\tt "http://example.edu/db/STUDENT\#NAME"}, \text{\tt "John"}) 
\end{align*}}

\vspace{-15pt}

\section{Properties of $\DM$}
\label{sec-lossless}
\noindent
We now study our direct mapping $\DM$ with respect to the
two fundamental properties (information preservation and query
preservation) and the two desirable properties (monotonicity and
semantics preservation) defined in Section \ref{sec-fp}. 

\subsection{Information preservation of $\DM$}
\noindent
First, we show that $\DM$ does not lose any piece of information in
the relational instance being translated:
\begin{theorem}\label{theo-dm-ip}
The direct mapping $\DM$ is information preserving.
\end{theorem}
The proof of this theorem is straightforward, and it involves providing
a computable mapping $\N : \rdf \to \ins$ that satisfies the condition
in Definition \ref{def-ip}, that is, a computable mapping $\N$ that
can reconstruct the initial relational instance from the generated RDF
graph.

\subsection{Query preservation of $\DM$}
\noindent
Second, we show that the way $\DM$ maps relational data into RDF
allows one to answer a query over a relational instance by translating
it into an equivalent query over the generated RDF graph.
\begin{theorem}\label{theo-dm-qp}
The direct mapping $\DM$ is query preserving.
\end{theorem}
In \cite{AG08}, it was proved that SPARQL has the same expressive
power as relational algebra. Thus, one may be tempted to think that
this result could be used to prove Theorem \ref{theo-dm-qp}. However,
the version of relational algebra considered in \cite{AG08} does not
include the null value $\nil$, and hence cannot be used to prove our
result.  In addition to this, other researchers have addressed the
issue of querying answering on DL ontologies with relational databases
\cite{SFD09}. Our work is similar in the sense that we address the
issue of query preservation between a database and an
ontology. However, the main difference is that rather than a domain
ontology, the ontology we use is synthesized in a standard way from
the database schema. Therefore, their results cannot be directly
applied to our setting.

We present an outline of the proof of this theorem, and refer the
reader to 
the Appendix for the details. Assume given a relational
schema $\R$ and a set $\Sigma$ of PKs and FKs over $\R$. Then we have to
show that for every relational algebra query $Q$ over $\R$, there
exists a SPARQL query $Q^\star$ such that for every instance $I$ of
$\R$ (possibly including null values) satisfying $\Sigma$:
\begin{eqnarray}
\label{eq-equiv}
\tr(\semr{Q}) & = & \semp{Q^\star}{\DM(\R,\Sigma,I)}.
\end{eqnarray}
\begin{sloppypar}
\noindent
Interestingly, the proof that the previous condition holds is by
induction on the structure of $Q$, and thus it gives us a bottom-up
algorithm for translating $Q$ into an equivalent SPARQL query
$Q^\star$, that is, a query $Q^\star$ satisfying condition
\eqref{eq-equiv}. In what follows, we consider the database used as
example in Section \ref{sec-dm} and the relational algebra
query $\sigma_{\text{\tt Name} = \text{\tt Juan}}(\text{\tt STUDENT})
\ \nj \ \text{\tt ENROLLED}$, which we will use as a running example and
translate it step by step to SPARQL, showing how the translation
algorithm works.
\end{sloppypar}

For the sake of readability, we introduce a function $\nu$ that
retrieves the IRI for a given relation $R$, denoted by $\nu(R)$, and
the IRI for a given attribute $A$ in a relation $R$, denoted by
$\nu(A,R)$. 
The inductive proof starts by considering the two base relational
algebra queries: the identity query $R$, where $R$ is a relation name
in the relational schema $\R$, and the query $\nil_A$. These two base
queries give rise to the following three base cases for the inductive
proof.


\smallskip

\begin{sloppypar}
\noindent
{\bf Non-binary relations:} Assume that $Q$ is the identity relational
algebra query $R$, where $R \in \R$ is a non-binary relation (that is,
$\IsBinRel(R)$ does not hold). Moreover, assume that $\att(R) = \{A_1,
\dots, A_\ell\}$, with the corresponding IRIs $\nu(R) = r, \nu(A_1,R)
= a_1, \ldots, \nu(A_\ell,R) = a_\ell$. Then a SPARQL query $Q^\star$
satisfying \eqref{eq-equiv} is constructed as follows:
\end{sloppypar}

\vspace*{-10pt}

\begin{align*}
\SELECT \ &\{?A_1, \ldots, ?A_\ell\} \ \bigg[ \cdots
\bigg(\bigg(\bigg((?X, \rdftype, r) \\  
&\opt (?X, a_1, ?A_1)\bigg) \opt (?X, a_2, ?A_2)\bigg)\\ 
&\opt (?X, a_3, ?A_3)\bigg) \cdots \opt (?X, a_\ell, ?A_\ell)\bigg].  
\end{align*}
\begin{sloppypar}
\noindent
Notice that in order to not lose information, the operator $\OPT$ is
used (instead of $\AND$) because the direct mapping $\DM$ does not
translate $\nil$ values. In our example, the relation name
{\tt STUDENT} is a non-binary relation. Therefore the following
equivalent SPARQL query is generated with input {\tt STUDENT}:
\end{sloppypar}

\vspace*{-10pt}

{\small
\begin{align*}
\SELECT \ &\{\text{\tt ?SID}, \text{\tt ?NAME} \} \ \bigg[
\bigg((?X, \rdftype, {\text{\tt :STUDENT}}) \\
& \opt (?X, {\text{\tt :STUDENT\#SID}}, \text{\tt ?SID})\bigg)\\ 
& \opt (?X, {\text{\tt :STUDENT\#NAME}}, \text{\tt ?NAME})\bigg] 
\end{align*}
\bw{It}} should be noticed that in the previous query, the symbol {\tt
:} has to be replaced by the base IRI used when generating IRIs for
relations and attributes in a relation (see Section
\ref{sec-iri}) \footnote{In SPARQL terminology, we have included the
following prefix in the query: {\tt
@prefix\,\,:\,\,<http://example.edu/db/>}, if the base IRI is 
{\tt <http://example.edu/db/>}.}.

\smallskip

\noindent
{\bf Binary relations:} Assume that $Q$ is the identity relational
algebra query $R$, where $R \in \R$ is a binary relation (that is,
$\IsBinRel(R)$ holds). Moreover, assume that $\att(R) = \{A_1, A_2\}$,
where $A_1$ is a foreign key referencing the attribute $B$ of a
relation $S$, and $A_2$ is a foreign key referencing the attribute $C$
of a relation $T$. Finally, assume that $\nu(R) = r$, $\nu(B,S) = b$
and $\nu(C,T) = c$, Then a SPARQL query $Q^\star$
satisfying \eqref{eq-equiv} is defined as follows:

\vspace*{-10pt}

\begin{multline*}
\SELECT \ \{?A_1, ?A_2\} \ ((?T_1, r, ?T_2) \ \andp\\ 
(?T_1, b,?A_1) \andp (?T_2, c, ?A_2)).
\end{multline*}
Given that a binary relation is mapped to an object property, the
values of a binary relation can be retrieved by querying the datatype
properties of the referenced attributes. In our example, the
relational name {\tt ENROLLED} is a binary relation. Therefore the
following equivalent SPARQL query is generated with input {\tt ENROLLED}:

\vspace*{-10pt}

{\small
\begin{align*}
\SELECT \ &\{\text{\tt ?SID}, \text{\tt ?CID}\} (\\
& (?T_1, {\text{\tt :ENROLLED\#SID,CID,SID,CID}}, ?T_2) \ \andp\\
& (?T_1, {\text{\tt :STUDENT\#SID}},\text{\tt ?SID}) \ \andp\\
& (?T_2, {\text{\tt :COURSE\#CID}}, \text{\tt ?CID})). 
\end{align*}
}

\noindent
{\bf Empty relation:} Assume that $Q = \nil_A$, and define $Q^\star$
as the empty graph pattern $\egp$. Then we have that condition
\eqref{eq-equiv} holds because of the definition of the function $\tr$,
which does not translate $\nil$ values to mappings.

\smallskip

We now present the inductive step in the proof of Theorem
\ref{theo-dm-qp}. Assume that the theorem holds for relational algebra
queries $\Qone$ and $\Qtwo$. That is, there exists SPARQL queries
$\Qstarone$ and $\Qstartwo$ such that:
\begin{eqnarray}
\label{eq-equiv-1}
\tr(\semr{\Qone}) & = & \semp{\Qstarone}{\DM(\R,\Sigma,I)},\\  
\label{eq-equiv-2}
\tr(\semr{\Qtwo}) & = & \semp{\Qstartwo}{\DM(\R,\Sigma,I)}.
\end{eqnarray}
The proof continues by presenting equivalent SPARQL queries for the
following relational algebra operators: selection ($\sigma$),
projection ($\pi$), rename ($\delta$), join ($\nj$), union ($\cup$)
and difference ($\smallsetminus$). It is important to notice that the
operators left-outer join, right-outer join and full-outer join are
all expressible with the previous operators, hence we do not present
cases for these operators.

\smallskip

\noindent{\bf Selection:} We need to consider four cases to define
query $Q^\star$ satisfying condition (\ref{eq-equiv}). In all these
cases, we use the already established equivalence \eqref{eq-equiv-1}.
\begin{alsl}
\item[1.] If $Q$ is $\sigma_{A_1 = a}(\Qone)$, then
\begin{eqnarray*}
Q^\star & = & (\Qstarone \vc (?A_1 = a)).
\end{eqnarray*}

\item[2.] If $Q$ is $\sigma_{A_1 \neq a}(\Qone)$, then 
\begin{eqnarray*}
Q^\star & = & (\Qstarone \vc (\neg(?A_1 = a) \wedge \bound(?A_1))).
\end{eqnarray*}

\item[3.] If $Q$ is $\sigma_{\isn(A_1)}(\Qone)$, then 
\begin{eqnarray*}
Q^\star & = & (\Qstarone \vc(\neg\bound(?A_1))).
\end{eqnarray*}

\item[4.] If $Q$ is $\sigma_{\isnn(A_1)}(\Qone)$, then 
\begin{eqnarray*}
Q^\star & = & (\Qstarone \vc(\bound(?A_1))).
\end{eqnarray*}
\end{alsl}
These equivalences are straightforward. However, it is important to
note the use of $\bound(\cdot)$ in the second case; as the semantics
of relational algebra states that if $Q$ is the query $\sigma_{A_1
\neq a}(\Qone)$, then $\semr{Q} = \{ t \in \semr{\Qone} \mid t.A_1
\neq \nil \text{ and } t.A_1 \neq a\}$, we have that the variable
$?A_1$ has to be bound because the values in the attribute $A_1$ in
the answer to $\sigma_{A_1 \neq a}(\Qone)$ are different from
$\nil$. Following our example, we have that the following
SPARQL query is generated with input $\sigma_{\text{\tt Name} =
\text{\tt Juan}}(\text{\tt STUDENT})$:

\vspace*{-10pt}

{\small
\begin{align*}
\bigg(\SELECT \ &\{\text{\tt ?SID}, \text{\tt ?NAME} \} \ \bigg[
\bigg((?X, \rdftype, {\text{\tt :STUDENT}}) \\
& \opt (?X, {\text{\tt :STUDENT\#SID}}, \text{\tt ?SID})\bigg)\\ 
& \opt (?X, {\text{\tt :STUDENT\#NAME}}, \text{\tt ?NAME})\bigg]\bigg)\\
& \hspace{100pt} \vc (\text{\tt ?NAME} = \text{\tt Juan})
\end{align*}}

\begin{sloppypar}
\noindent
{\bf Projection:} Assume that $Q = \pi_{\{A_1, \dots,
A_\ell\}}(\Qone)$. Then query $Q^\star$ satisfying condition
\eqref{eq-equiv} is defined as $\sel{\{?A_1,
\ldots,?A_\ell\}}{\Qstarone}$. It is important to notice that we use
nested $\SELECT$ queries to deal with projection, as well as in two of
the base cases, which is a functionality specific to SPARQL 1.1~\cite{HS11}.
\end{sloppypar}

\smallskip

\noindent{\bf Rename:}  Assume that $Q = \delta_{A_1 \to B_1}(\Qone)$
and $\att(Q) = \{A_1$, $\ldots$, $A_\ell\}$. Then query $Q^\star$
satisfying condition \eqref{eq-equiv} is defined as $\sel{\{?A_1 \as
?B_1, ?A_2, \dots , ?A_\ell\}}{\Qstarone}$. Notice that this
equivalence holds because the rename operator in relational algebra
renames one attribute to another and projects all attributes of $Q$.
 
\smallskip

\noindent{\bf Join:}  Assume that $Q = (Q_1 \nj Q_2)$, where
$(\att(Q_1) \cap \att(Q_2)) = \{A_1, \ldots, A_\ell\}$. Then query
$Q^\star$ satisfying condition \eqref{eq-equiv} is defined as follows:
\begin{multline*}
\bigg[\bigg(Q_1^\star \vc (\bound(?A_1) \wedge \cdots
\wedge \bound(?A_\ell))\bigg) \ \andp \\ 
\bigg(Q_2^\star \vc (\bound(?A_1) \wedge \cdots \wedge
\bound(?A_\ell))\bigg)\bigg]. 
\end{multline*}
\begin{sloppypar}
\noindent
Note the use of $\bound(\cdot)$ which is necessary in the SPARQL query
in order to guarantee that the variables that are being joined on are
not null. Following our example, Figure \ref{fig-trans} shows
the SPARQL query generated with input $\sigma_{\text{\tt Name} =
\text{\tt Juan}}(\text{\tt STUDENT}) \ \nj \ \text{\tt ENROLLED}$.
\end{sloppypar}

\begin{figure*}
{\small
\begin{align*}
&\bigg[\bigg(\bigg(\SELECT \ \{\text{\tt ?SID}, \text{\tt ?NAME} \} \
\bigg[\bigg((?X, \rdftype, {\text{\tt :STUDENT}}) \opt (?X, {\text{\tt
:STUDENT\#SID}}, \text{\tt ?SID})\bigg) \ \opt\\ 
&\hspace{150pt}(?X, {\text{\tt :STUDENT\#NAME}}, \text{\tt
?NAME})\bigg]\bigg) \vc (\text{\tt ?NAME} = \text{\tt Juan})\bigg) \vc
(\bound(\text{\tt ?SID}))\bigg]\\
&\hspace{200pt} \andp\\
&\bigg[\bigg(\SELECT \ \{\text{\tt ?SID}, \text{\tt ?CID}\}
\bigg((?T_1, {\text{\tt :ENROLLED\#SID,CID,SID,CID}}, ?T_2) \andp (?T_1, {\text{\tt
:STUDENT\#SID}},\text{\tt  ?SID}) \ \andp\\ 
& \hspace{263pt} (?T_2, {\text{\tt :COURSE\#CID}}, \text{\tt ?CID})\bigg)\bigg)
\vc (\bound(\text{\tt ?SID}))\bigg]
\end{align*}
}

\vspace*{-14pt}

\caption{SPARQL translation of the relational algebra query
$\sigma_{\text{\tt Name} = \text{\tt Juan}}(\text{\tt STUDENT}) \ \nj
\ \text{\tt ENROLLED}$. \label{fig-trans}}
\end{figure*}

\smallskip

\noindent{\bf Union:} Assume that $Q = (Q_1 \cup Q_2)$. Then query $Q^\star$
satisfying condition \eqref{eq-equiv} is simply defined as $(Q_1^\star
\uni Q_2^\star)$. Notice that in this case we are using the already
established equivalences \eqref{eq-equiv-1} and~\eqref{eq-equiv-2}.

\smallskip

\noindent{\bf Difference:} We conclude our proof by assuming that $Q =
(Q_1 \smallsetminus Q_2)$. In this case, it is also possible to define
a SPARQL query $Q^\star$ satisfying condition (\ref{eq-equiv}). 
We refer the reader to 
the appendix for the complete description of $Q^\star$.

\subsection{Monotonicity and semantics preservation of $\DM$}
\noindent
Finally, we consider the two desirable properties identified in
Section \ref{sec-desi}. First, it is straightforward to see that $\DM$
is monotone, because all the negative atoms in the Datalog rules
defining $\DM$ refer to the schema, the PKs and the FKs of the
database, and these elements are kept fixed when checking
monotonicity. Unfortunately, the situation is completely different for
the case of semantics preservation, as the following example shows
that the direct mapping $\DM$ does not satisfy this property.
\begin{example}\label{sp-example} {\em Assume that
a relational schema contains a relation with name {\tt STUDENT} and
attributes {\tt SID}, {\tt NAME}, and assume that the attribute {\tt
SID} is the primary key.  Moreover, assume that this relation has two
tuples, $t_1$ and $t_2$ such that $t_1.\text{\tt SID}= \text{1}$,
$t_1.\text{\tt NAME} = \text{John}$ and $t_2.\text{\tt SID} =
\text{1}$, $t_2.\text{\tt NAME} = \text{Peter}$. It is clear that the
primary key is violated, therefore the database is
inconsistent. However, it is not difficult to see that after applying
$\DM$, the resulting RDF graph is consistent. \qed}
\end{example}
In fact, the result in Example \ref{sp-example} can be generalized as
it is possible to show that the direct mapping $\DM$ always generates
a consistent RDF graph, hence, it cannot be semantics preserving.
\begin{proposition}\label{prop-dm-sp}
The direct mapping $\DM$ is not semantics preserving.
\end{proposition}
Does this mean that our direct mapping is incorrect? What could we do
to create a direct mapping that is semantics preserving? These problems
are studied in depth in the following section.

\section{Semantics Preservation of\\ Direct Mappings}
\noindent
We now study the problem of generating a
semantics-preserving direct mapping. Specifically, we show in Section
\ref{sec-pk} that a simple extension of the direct mapping $\DM$ can
deal with primary keys. Then we show in Section \ref{sec-pk-fk} that
dealing with foreign keys is more difficult, as any direct mapping
that satisfies the condition of being monotone cannot be
semantics preserving. Finally, we present
two possible ways of overcoming this limitation.

\subsection{A semantics preserving direct mapping for primary keys}
\label{sec-pk}
\noindent
Recall that a primary key can be violated if there are repeated values or null values.
At a first glance, one would assume that owl:hasKey could be used to 
create a semantics preserving direct mapping for primary keys.  
If we consider a database without null values, a violation of the primary key would
generate an inconsistency with owl:hasKey and the unique name assumption (UNA). 
However, if we consider a database with null values, 
then owl:hasKey with the UNA does not generate an inconsistency because it is trivially satisfied for a class expression that does not have a value for the datatype expression. Therefore, we must consider a different approach.

Consider a new direct mapping $\DM_\text{\it pk}$ that extends $\DM$
as follows. A Datalog rule is used to determine if the value of a
primary key attribute is repeated, and a family of Datalog rules are
used to determine if there is a value $\nil$ in a column corresponding
to a primary key. If some of these violations are found, then an
artificial triple is generated that would produce an
inconsistency. For example, the following rules are used to map a
primary key with two attributes:

\vspace*{-10pt}

{\small
\begin{align*}
\Triple(a, \text{\tt "o}&\text{\tt wl:differentFrom"}, a) \leftarrow
\PK_2(X_1, X_2, R),\\   
&\Value(V_1, X_1, T_1, R), \Value(V_1, X_1, T_2, R),\\
&\Value(V_2, X_2, T_1, R), \Value(V_2, X_2, T_2, R), T_1 \neq T_2\\
\Triple(a, \text{\tt "o}&\text{\tt wl:differentFrom"}, a) \leftarrow
\PK_2(X_1, X_2, R),\\  
&\Value(V, X_1, T, R), V = \nil
\end{align*}
\begin{align*}
\Triple(a, \text{\tt "o}&\text{\tt wl:differentFrom"}, a) \leftarrow
\PK_2(X_1, X_2, R),\\  
&\Value(V, X_2, T, R), V = \nil
\end{align*}
\bw{In}} the previous rules, $a$ is any valid IRI. If we apply
$\DM_\text{\it pk}$ to the database of Example \ref{sp-example}, it is
straightforward to see that starting from an inconsistent relational
database, one obtains an RDF graph that is also inconsistent. In fact,
we have that:
\begin{sloppypar}
\begin{proposition}\label{prop-dmpk-sp}
The direct mapping $\DM_\text{\it pk}$ is information preserving, query
preserving, monotone, and semantics preserving if one considers only
PKs. That is, for every relational schema $\R$, set $\Sigma$ of (only)
PKs over $\R$ and instance $I$ of $\R$:
$I \models \Sigma$ \ iff \ $\DM_\text{\it pk}(\R, \Sigma, I)$ is
consistent under OWL semantics.  
\end{proposition}
\end{sloppypar}
\noindent
Information preservation, query preservation and monotonicity of
$\DM_\text{\it pk}$ are corollaries of the fact that these properties
hold for $\DM$, and of the fact that the Datalog rules introduced to
handle primary keys are monotone. 

A natural question at this point is whether $\DM_{\text{\it pk}}$ can
also deal with foreign keys. Unfortunately, it is easy to construct an
example that shows that this is not the case. Does this mean that we
cannot have a direct mapping that is semantics preserving and
considers foreign keys?  We show in the following section that
monotonicity has been one of the obstacles to obtain such a mapping.





\subsection{Semantics preserving direct mappings for primary keys and
foreign keys}
\label{sec-pk-fk}
\noindent
The following theorem shows that the desirable condition of being
monotone is, unfortunately, an obstacle to obtain a semantics
preserving direct mapping.
\begin{theorem}\label{theo-mono-map-no}
No monotone direct mapping is semantics preserving.
\end{theorem}
It is important to understand the reasons why we have not been able to
create a semantics preserving direct mapping. The issue is with 
two characteristics of OWL: (1) it adopts the Open World
Assumption (OWA), where a statement cannot be inferred to be false on
the basis of failing to prove it, and (2) it does not adopt the Unique
Name Assumption (UNA), where two different names can identify the same
thing. On the other hand, a relational database adopts the Closed
World Assumption (CWA), where a statement is inferred to be false if
it is not known to be true.
In other
words, what causes an inconsistency in a relational database, can
cause an inference of new knowledge in OWL.

In order to preserve the semantics of the relational database, we need
to ensure that whatever causes an inconsistency in a relational
database, is going to cause an inconsistency in OWL.  
Following this idea, we now present a non-monotone direct mapping,
$\DM_{\text{\it pk}+\text{\it fk}}$, which extends $\DM_\text{\it pk}$
by introducing rules for verifying beforehand if there is a violation
of a foreign key constraint. If such a violation exists, then an
artificial RDF triple is created which will generate an inconsistency
with respect to the OWL semantics. More precisely, the following
family of Datalog rules are used in $\DM_{\text{\it pk}+\text{\it
fk}}$ to detect an inconsistency in a relational database:

\vspace*{-10pt}

{\small
\begin{align*}
\Violation(S) & \leftarrow\\
& \FK_n(X_1, \ldots, X_n, S, Y_1, \ldots, Y_n, T),\\ 
& \Value_n(V_1, X_1, T, S), \ldots, \Value(V_n, X_n, T, S),\\
& V_1 \neq \nil, \ldots,  V_n \neq \nil, \\
& \neg \IsValue_n(V_1, \ldots, V_n, Y_1, \ldots, Y_n,T)
\end{align*}
\bw{In}} the preceding rule, the predicate $\IsValue_n$ is used to
check whether a tuple in a relation has values for some given
attributes. The predicate $\IsValue_n$ is defined by the following
rule:

\vspace*{-10pt}

{\small
\begin{multline*}
\IsValue_n(V_1, \ldots, V_n,B_1, \ldots, B_n,S) \ \leftarrow \\ 
\Value(V_1, B_1, T, S), \ldots, \Value(V_n, B_n, T, S) 
\end{multline*}
\bw{Finally,}} the following Datalog rule is used to obtain
an inconsistency in the generated RDF graph:

\vspace*{-10pt}

{\small
\begin{eqnarray*}
\Triple(a, \owldf, a) & \leftarrow & \Violation(S)
\end{eqnarray*}
\bw{In}} the previous rule,  $a$ is any valid IRI. It should be noticed
that $\DM_{\text{\it pk}+\text{\it fk}}$ is non-monotone because if
new data in the database is added which now satisfies the FK
constraint, then the artificial RDF triple needs to be retracted. 

\begin{theorem}\label{theo-dm-sp-pk-fk}
The direct mapping $\DM_{\text{\it pk} + \text{\it fk}}$ is
information preserving, query preserving and semantics
preserving. 
\end{theorem}
Information preservation and query preservation of $\DM_{\text{\it pk}
+ \text{\it fk}}$ are corollaries of the fact that these properties
hold for $\DM$ and $\DM_\text{\it pk}$.

A direct mapping that satisfies the four properties can be obtained by considering an 
alternative semantics of OWL that expresses
integrity constraints. Because OWL is based on Description Logic, we
would need a version of DL that supports integrity constraints, which
is not a new idea. Integrity constraints are epistemic in nature and
are about ``what the knowledge base knows''~\cite{Rei88}. Extending DL
with the epistemic operator {\bf K} has been studied
\cite{CGL+07,DLN+98,DNR02}.  Grimm et al. proposed to extend the
semantics of OWL to support the epistemic operator
\cite{GM05}. Motik et al. proposed to write integrity constraints as
standard OWL axioms but interpreted with different semantics for 
data validation purposes \cite{MHS09}. Tao et al. showed that integrity
constraint validation can be reduced to SPARQL query answering
\cite{TSB+10}. Recently, Mehdi et al. introduced a way to answer
epistemic queries to restricted OWL ontologies
\cite{MRG11}. Thus, it is possible to extend $\DM_\text{\it pk}$
to create an information preserving, query preserving and monotone
direct mapping that is also semantics preserving, but it is based on a
non-standard version of OWL including the epistemic operator {\bf K}.


\section{Concluding remarks}
\noindent
In this paper, we study how to directly map relational databases to an
RDF graph with OWL vocabulary based on two fundamental
properties (information preservation and query preservation) and two desirable properties (monotonicity 
and semantics preservation). We first present a monotone, information
preserving 
and query preserving direct
mapping considering databases that have null values. Then we prove that the combination of monotonicity with the
OWL semantics is an obstacle to generating a semantics preserving
direct mapping. Finally, we overcome this obstacle by presenting a
non-monotone direct mapping that is semantics preserving, and also by
discussing the possibility of generating a monotone mapping that
assumes an extension of OWL with the epistemic operator.

\smallskip

\noindent
{\bf Related Work:} Several approaches directly map relational 
schemas to RDFS and OWL. We refer the reader to the following survey
\cite{STC+12}. D2R Server has an option that directly maps the
relational database into RDF, however this process is not documented
\cite{D2R}. RDBToOnto presents a direct mapping that mines the content
of the relational databases in order to learn ontologies with deeper
taxonomies \cite{Cer08}.  Currently, the W3C RDB2RDF Working Group is
developing a direct mapping standard that focuses on translating
relational database instances to RDF \cite{APS11,BP11}.  

\smallskip

\noindent
{\bf Future Work:} We would like to extend our direct mapping to
consider datatypes, relational databases under bag semantics and
evaluate this rule based approach on large relational databases.  The
extension of our direct mapping to bag semantics is straightforward.
In our setting each tuple has its own identifier, which is
represented in the $\Value$ predicate. Thus, even if repeated tuples
exist, each tuple will still have its unique identifier and,
therefore, exactly the same rules can be used to map relational data
under bag semantics.

\smallskip


\noindent
\section{Acknowledgments}
\noindent
The authors would like to thank the anonymous referees for many
helpful comments, and the members of the W3C RDB2RDF Working group for
many fruitful discussions. J. F. Sequeda was supported by the NSF
Graduate Research Fellowship, M. Arenas by Fondecyt grant
\#1090565 and D.P. Miranker by NSF grant \#1018554.

\newpage
\onecolumn

\appendix

\renewcommand{\refapp}{{ADDITIONAL REFERENCES FOR THE APPENDIX}}


\section{Additional Operators in Relational algebra}
\label{sec-ra-app}
\noindent
It is important to notice that the operators left-outer join,
right-outer join and full-outer join are all
expressible with the previous operators. For example, assume that $R$
and $S$ are relation names such that $\att(R) \cap \att(S) =
\{A_1, A_2, \ldots, A_k\}$ and $\att(S) \smallsetminus \att(R) = \{B_1,
B_2, \ldots, B_\ell\}$, then the left-outer join for $R$ and $S$ is
defined by the following expression:  
{\small
\begin{align*}
\bigg[R \nj S\bigg] \cup \bigg[&\sigma_{\isn(A_1)}(R) \cup
\sigma_{\isn(A_2)}(R) \cup \cdots \cup \sigma_{\isn(A_k)}(R) \ \cup \\
& R \nj
\bigg(\sigma_{\isnn(A_1)}(\sigma_{\isnn(A_2)}(\cdots\sigma_{\isnn(A_k)}(\pi_{\{A_1,A_2,\ldots,A_k\}}(R))\cdots))
\smallsetminus \pi_{\{A_1,A_2,\ldots,A_k\}}(S)\bigg)\bigg] \nj\\
& \hspace{250pt} \bigg[\nil_{B_1} \nj \nil_{B_2} \nj \cdots \nj
\nil_{B_\ell}\bigg]. 
\end{align*}}
Similar expressions can be used to express the right-outer join and
the full-outer join.


\section{Semantics of SPARQL}
\label{sec-sparql-app}
\noindent
Let $P$ be a SPARQL graph pattern.  In the rest of the paper, we use
$\var(P)$ to denote the set of variables occurring in $P$. In
particular, if $t$ is a triple pattern, then $\var(t)$ denotes the set
of variables occurring in the components of $t$.  Similarly, for a
built-in condition $R$, we use $\var(R)$ to denote the set of
variables occurring in $R$. 

In what
follows, we present the semantics of graph patterns for a fragment of
SPARQL for which the semantics of nested SELECT queries is well
understood~\cite{HS11,AG10,AG11}. More specifically, in what follows
we focus on the class of graph patterns $P$ satisfying the following
condition: $P$ is said to be {\em non-parametric} if for every
sub-pattern $P_ 1 = \sel{\{?A_1 \as ?B_1, \dots , ?A_m \as ?B_m, ?C_1,
\dots , ?C_n\}}{P_2}$ of $P$ and every variable $?X$ occurring in $P$,
if $?X \in (\var(P_2) \smallsetminus \{?A_1, \ldots, ?A_m, ?C_1,
\ldots, ?C_n\})$, then $?X$ does not occur in
$P$ outside $P_1$.

To define the semantics of SPARQL graph pattern expressions, we need
to introduce some terminology.  A mapping $\mu$ is a partial function
$\mu: \V \rightarrow (\I \cup \L)$.  Abusing notation, for a triple
pattern $t$ we denote by $\mu(t)$ the triple obtained by replacing the
variables in $t$ according to $\mu$.  The domain of $\mu$, denoted by
$\dom(\mu)$, is the subset of $\V$ where $\mu$ is defined. Two
mappings $\mu_1$ and $\mu_2$ are compatible, denoted by $\mu_1 \sim
\mu_2$, when for all
$x\in\dom(\mu_1)\cap\dom(\mu_2)$, it is the case that
$\mu_1(x)=\mu_2(x)$, i.e. when $\mu_1\cup\mu_2$ is also a mapping.
The mapping with empty domain is denoted by $\mu_\emptyset$ (notice
that this mapping is compatible with any other mapping).  Given a
mapping $\mu$ and a set of variables $W$, the restriction of $\mu$ to
$W$, denoted by $\res{\mu}{W}$, is a mapping such that
$\dom(\res{\mu}{W})=\dom(\mu)\cap W$ and $\res{\mu}{W}(?X)=\mu(?X)$
for every $?X\in \dom(\mu)\cap W$. Finally, given a mapping $\mu$ and
a sequence $?A_1$, $\ldots$, $?A_m$, $?B_1$, $\ldots$, $?B_m$ of
pairwise distinct elements from $\V$ such that $\dom(\mu) \cap \{?B_1,
\ldots, ?B_m\} = \emptyset$, define $\ren{?A_1 \to ?B_1,
\ldots, ?A_m \to ?B_m}(\mu)$ as a mapping such that: 
\begin{multline*}
\dom(\ren{?A_1 \to ?B_1, \dots, ?A_m \to ?B_m}(\mu)) \ =  (\dom(\mu)
\smallsetminus \{?A_1, \dots, ?A_m\}) \cup \{?B_i \mid i \in \{1,
\ldots, m\} \text{ and } ?A_i \in \dom(\mu)\},
\end{multline*}
and for every $x \in \dom(\ren{?A_1 \to ?B_1, \dots, ?A_m \to
?B_m}(\mu))$: 
\begin{equation*}
\ren{?A_1 \to ?B_1, \dots, ?A_m \to ?B_m}(\mu)(x) = 
\begin{cases}
\mu(?A_i) & x = ?B_i \text{ for some } i \in \{1, \ldots, m\}\\
\mu(x) & \text{otherwise}
\end{cases}
\end{equation*}
We have all the necessary ingredients to define the semantics of graph
pattern expressions. As in \cite{PAG09}, we define this semantics as a
function $\sem{\,\cdot\,}$ that takes a graph pattern expression and
returns a set of mappings. For the sake of readability, the semantics
of filter expressions is presented separately.

The evaluation of a graph pattern $P$ over an RDF graph $G$, denoted
by $\sem{P}$, is defined recursively as follows.
\begin{enumerate}
\item If $P$ is $\egp$ and $G$ is nonempty, then $\sem{P} = \{
\mu_\emptyset \}$. If $P$ is $\egp$ and $G = \emptyset$, then $\sem{P}
= \emptyset$.

\item If $P$ is a triple pattern $t$, then
$\sem{P}=\{\mu \mid \dom(\mu)=\var(t)$ and $\mu(t)\in G\}$.

\item If $P$ is  $(P_1 \andp P_2)$, then $\sem{P} = \{\mu_1 \cup \mu_2
\mid \mu_1 \in \sem{P_1},\ \mu_2 \in \sem{P_2}$ and $\mu_1 \sim \mu_2\}$.

\item If $P$ is $(P_1 \opt P_2)$, then $\sem{P} = \{\mu_1 \cup \mu_2
\mid \mu_1 \in \sem{P_1},\ \mu_2 \in \sem{P_2}$ and $\mu_1 \sim
\mu_2\} \cup \{\mu \in \sem{P_1} \mid \text{ for every } \mu'\in
\sem{P_2}:  \mu \not\sim \mu'\}$. 

\item If $P$ is  $(P_1 \uni P_2)$, then $\sem{P} = \{\mu \mid \mu \in
\sem{P_1}$ or $\mu \in \sem{P_2}\}$.

\item If $P$ is  $(P_1 \minus P_2)$, then $\sem{P} = \{\mu \in \sem{P_1}
\mid \text{ for every } \mu'\in  \sem{P_2}:  \mu \not\sim \mu'$ or
$\dom(\mu) \cap \dom(\mu') = \emptyset\}$.    

\item If $P$ is $\sel{\{?A_1 \as ?B_1, \dots , ?A_m \as ?B_m, ?C_1,
\dots , ?C_n\}}{P_1}$, then:
\begin{eqnarray*} 
\sem{P} & = & \{ \ren{?A_1 \to ?B_1, \dots, ?A_m \to
?B_m}(\res{\mu}{\{?A_1, \dots, ?A_m, ?C_1, \dots, ?C_n\}}) 
\mid \mu \in \sem{P_1}\}.
\end{eqnarray*} 
\end{enumerate}
The semantics of filter expressions goes as follows.  Given a mapping
$\mu$ and a built-in condition $R$, we say that $\mu$ satisfies $R$,
denoted by $\mu \models R$, if:
\begin{enumerate}
\item $R$ is $\bound(?X)$ and $?X \in \dom(\mu)$;

\item $R$ is $?X = c$, $?X \in \dom(\mu)$ and $\mu(?X) = c$;

\item $R$ is $?X = ?Y$, $?X \in \dom(\mu)$, $?Y \in \dom(\mu)$ and $\mu(?X) =
\mu(?Y)$;

\item $R$ is $(\neg R_1)$, $R_1$ is a built-in condition, and it
is not the case that $\mu \models R_1$;

\item $R$ is $(R_1 \vee R_2)$,  $R_1$ and $R_2$ are built-in
conditions, and $\mu \models R_1$ or $\mu \models R_2$;

\item $R$ is $(R_1 \wedge R_2)$, $R_1$ and $R_2$ are built-in
conditions, $\mu \models R_1$ and $\mu \models R_2$.
\end{enumerate}
Then given an RDF graph $G$ and a filter expression $(P \vc R)$:
\begin{eqnarray*}
\sem{(P \vc R)} & = & \{ \mu \in \sem{P} \mid \mu \models R\}.
\end{eqnarray*}

\section{Proofs}


\subsection{Proof of Theorem \ref{theo-dm-ip}}
We show that $\DM$ is information preserving by providing a computable
mapping $\N : \rdf \to {\cal I}$ that satisfies the condition in
Definition
\ref{def-ip}. More precisely, given a relational schema $\R$, a set
$\Sigma$ of PKs and FKs and an instance $I$ of $\R$ satisfying
$\Sigma$, next we should how $\N(G)$ is defined for $\DM(\R, \Sigma, I) = G$.

\begin{itemize}
\item {\bf Step 1:} Identify all the ontological class triples (i.e
\Triple($r$, \rdftype, \owlClass)). The IRI $r$ identifies an
ontological class $R'$. For every $R'$ that was retrieved from
$G$, map it to a relation name $R$.

\item \begin{sloppypar} {\bf Step 2:} Identify all the datatype
triples of a given class (i.e \Triple($a$, \rdftype, \owlDTP),
\Triple($a$, \domain, $r_i$)). The IRI $a$ identifies the datatype
property $A'$ and the IRI $r$ identifies the ontological class
$R'$ that is the domain of $A'$. Every datatype property $A'$
with domain $R'$ is mapped to an attribute $A$ of relation name
$R$. \end{sloppypar}

\item {\bf Step 3:} For each class $R'$ and the datatype properties
$A_1' \ldots A_n'$ that have domain $R'$, we can recover the instances
of relation $R$ with the following SPARQL query:

\begin{multline*}
Q_1 \ = \  \SELECT \ \{?A_1, \ldots, ?A_n\} \ \bigg[ \cdots
\bigg(\bigg(\bigg((?X, \rdftype, r_i) \opt (?X, a_1, ?A_1)\bigg) \opt
(?X, a_2, ?A_2)\bigg)\\ \opt (?X, a_3, ?A_3)\bigg) \cdots \opt (?X_,
a_n, ?A_n)\bigg].  
\end{multline*}

\item \begin{sloppypar} {\bf Step 4:} Identify all the object property
triples (i.e. \Triple($r$, \rdftype, \owlOP)). The IRI $r$ that only has one element left of the $\#$ sign 
means that $r$
identifies the object property $R'$ in the
ontology that was originally mapped from a binary relation. This
object property $R'$ is mapped back to a binary relation name
$R$. The two elements following the $\#$ sign identify the attributes of the relation $R$.
From the triples
\Triple(r, \domain, s) and \Triple(r, \range, t), the IRI $s$
identifies the ontological class $S'$ which is mapped to the relation
$S$ and the IRI $t$ identifies the ontological class $T'$ which is
mapped to the relation $T$. Additionally,
the elements in the third and fourth position after the $\#$ identify the attributes which are being referenced from 
relations $S$ and $T$ respectively. For sake of simplicity, assume that the relation $R$ references the attribute $B$ of relation $S$ which is mapped to a datatype property $B'$ with domain $S'$ and IRI $b$. Additionally, the relation $R$ references the attribute $C$ of relation $T$, which is mapped to a datatype property $C'$
with domain $T'$ and IRI $c$. 

We can now recover the instances of the
relation $R$ with the following SPARQL query:
\begin{eqnarray*}
Q^\star & = & \sel{\{?A_1, ?A_2\}}{((?T_1, r, ?T_2) \andp (?T_1, b,
?A_1) \andp (?T_2, c, ?A_2))}.
\end{eqnarray*}
\end{sloppypar}

\item {\bf Step 5:} Given that the result of a SPARQL query is a set
$\Omega$ of solution mapping $\mu$, we need to translate each solution
mapping $\mu \in \Omega$ into a tuple $t$. We define a function
$\tr^{-1}$ as the inverse of function $\tr$, that is, for each
solution mapping $\mu$ and variable $?A$ in the domain of $\mu$,
$\tr^{-1}$ assigns the value of $\mu(?A)$ to $t.A$. 
Then the mapping function $\N$ over $G$ is defined as the following
relational instance. For every non-binary relation name identified in Steps 1, 2,
3, define $R^{\N(G)}$ as $\tr^{-1}(\semp{Q_1}{G})$, and for every binary
relation $R$ identified in Step 4, define $R^{\N(G)}$ as
$\tr^{-1}(\semp{Q_2}{G})$.
\end{itemize}
It is straightforward to prove that for every relational schema $\R$,
set $\Sigma$ of PKs and FKs and an instance $I$ of $\R$ satisfying
$\Sigma$, it holds that $\N(\M(\R, \Sigma, I)) = I$. This concludes
the proof of the theorem.


\subsection{Proof of Theorem \ref{theo-dm-qp}}

We need to prove that for every relational schema $\R$, set $\Sigma$
of PKs and FKs over $\R$ and relational algebra query $Q$ over $\R$,
there exists a SPARQL query $Q^\star$ such that for every instance $I$
of $\R$ including null values:
\begin{eqnarray*}
\tr(\semr{Q}) & = & \semp{Q^\star}{\DM(\R,\Sigma,I)}.
\end{eqnarray*}
In what follows, assume that $\R$ is a relational schema, $\Sigma$ is
a set of PKs and FKs over $\R$, and $I$ is an instance of $\R$
satisfying $\Sigma$. The following lemma is used in the proof of the
theorem.
\begin{lemma}\label{lem-sel}
Let $Q_1$ be a relational algebra query over $\R$ such that $\att(Q_1)
= \{A_1, \ldots, A_\ell\}$, and assume that $Q^\star_1$ is a SPARQL
graph pattern such that:
\begin{eqnarray*}
\tr(\semr{Q_1}) & = & \semp{Q_1^\star}{\DM(\R,\Sigma,I)}.
\end{eqnarray*}
Then we have that:
\begin{eqnarray*}
\tr(\semr{Q_1}) & = & \semp{\sel{\{?A_1, \ldots,
?A_\ell\}}{Q_1^\star}}{\DM(\R,\Sigma,I)}. 
\end{eqnarray*}
\end{lemma}
\begin{proof}
\begin{sloppypar}
First, we prove that $\tr(\semr{Q_1}) \subseteq \semp{\sel{\{?A_1,
\ldots, ?A_\ell\}}{Q_1^\star}}{\DM(\R,\Sigma,I)}$. Assume that $\mu
\in \tr(\semr{Q_1})$. Then there exists a tuple $t \in \semr{Q_1}$
such that $\tr(t) = \mu$. Thus, given that $\att(Q_1) = \{A_1, \ldots,
A_\ell\}$, we conclude that $\dom(\mu) \subseteq \{?A_1, \ldots,
?A_\ell\}$. Given that $\tr(\semr{Q_1}) =
\semp{Q_1^\star}{\DM(\R,\Sigma,I)}$, we have that $\mu \in
\semp{Q_1^\star}{\DM(\R,\Sigma,I)}$. Hence, from the fact that
$\dom(\mu) \subseteq \{?A_1, \ldots, ?A_\ell\}$, we conclude that $\mu
\in \semp{\sel{\{?A_1, \ldots,
?A_\ell\}}{Q_1^\star}}{\DM(\R,\Sigma,I)}$.

Second, we prove that $\semp{\sel{\{?A_1, \ldots,
?A_\ell\}}{Q_1^\star}}{\DM(\R,\Sigma,I)} \subseteq
\tr(\semr{Q_1})$. Assume that $\mu \in \semp{\sel{\{?A_1, \ldots,
?A_\ell\}}{Q_1^\star}}{\DM(\R,\Sigma,I)}$. Then there exists a mapping
$\mu' \in \semp{Q_1^\star}{\DM(\R,\Sigma,I)}$ such that $\mu =
\res{\mu'}{\{?A_1, \ldots, ?A_\ell\}}$. From the fact that 
$\tr(\semr{Q_1}) = \semp{Q_1^\star}{\DM(\R,\Sigma,I)}$, we conclude
that $\mu' \in \tr(\semr{Q_1})$. Thus, there exists a tuple $t \in
\semr{Q_1}$ such that $\tr(t) = \mu'$. But then given that $\att(Q_1)
= \{A_1, \ldots, A_\ell\}$, we conclude by definition of $\tr$ that
$\dom(\mu') \subseteq \{?A_1, \ldots, ?A_\ell\}$. Therefore, given
that $\mu = \res{\mu'}{\{?A_1, \ldots, ?A_\ell\}}$, we have that $\mu
= \mu'$ and, hence, $\mu \in \tr(\semr{Q_1})$ since $\mu' \in
\tr(\semr{Q_1})$. 
\end{sloppypar}
\end{proof}

\noindent
We now prove the theorem by induction on the structure of relational
algebra query $Q$.

\medskip

\noindent
{\bf Base Case:} For the sake of readability, we introduce a function
$\nu$ that retrieves the IRI for a given relation $R$, denoted by
$\nu(R)$, and the IRI for a given attribute $A$ in a relation $R$,
denoted by $\nu(A,R)$. In this part of the proof, we need to consider
three cases.
\begin{itemize}
\item {\bf Non-binary relations:} Assume that $Q$ is the identity
relational algebra query $R$, where $R$ is a non-binary relation
according to the definition given in Section \ref{sec-ro}. Moreover,
assume that $\att(R) = \{A_1, \dots, A_\ell\}$, with the corresponding
IRIs $\nu(R) = r, \nu(A_1,R) = a_1, \ldots, \nu(A_\ell,R) =
a_\ell$. Finally, let $Q^\star$ be the following SPARQL query:
\begin{multline*}
Q^\star \ = \  \SELECT \ \{?A_1, \ldots, ?A_\ell\} \ \bigg[ \cdots
\bigg(\bigg(\bigg((?X, \rdftype, r) \opt (?X, a_1, ?A_1)\bigg) \opt
(?X, a_2, ?A_2)\bigg)\\ \opt (?X, a_3, ?A_3)\bigg) \cdots \opt (?X,
a_\ell, ?A_\ell)\bigg].  
\end{multline*}
Next we prove that $\tr(\semr{Q}) = \semp{Q^\star}{\DM(\R,\Sigma,I)}$. 

First, we show that $\tr(\semr{Q}) \subseteq
\semp{Q^\star}{\DM(\R,\Sigma,I)}$. Assume that $\mu \in
\tr(\semr{Q})$. Then there exists a tuple $t \in \semr{Q}$ such that
$\tr(t) = \mu$ and, hence, $t \in R^I$. Without loss of generality,
assume that there exists $k \in \{0, \ldots, \ell\}$ such that (1)
$t.A_i \neq \nil$ for every $i \in \{1, \ldots k\}$, and (2) $t.A_j =
\nil$ for every $j \in \{k+1, \ldots, \ell\}$. By definition of $\tr$,
we have that $t.A_i = \mu(?A_i)$ for every $i \in \{1, \ldots, k\}$,
and that $\dom(\mu) = \{?A_1, \dots, ?A_k\}$. Given the definition of
$\DM$, we have that the following holds: $\Class(R)$ and $\DTP(A_i,
R)$ for every $i \in \{1, \ldots, \ell\}$.  Hence, given that $R$ is
not a binary relation (that is, $\IsBinRel(R)$ does not hold), we have
that the following triples are included in $\DM(\R,\Sigma,I)$:
\begin{itemize}
\item $(r_{{\it id}}, \rdftype, r)$, where $r_{{\it id}}$ is
the tuple id for the tuple $t$, and

\item $(r_{{\it id}}, a_i, v_i)$, where $i \in \{1, \ldots,
k\}$ and $v_i$ is the value of attribute $A_i$ in the tuple $t$, that
is, $t.A_i = v_i$.
\end{itemize}
Thus, given that no triple of the form $(r_{{\it id}}, a_j, v_j)$ is
included in $\DM(\R,\Sigma,I)$, for $j \in \{k+1, \ldots, \ell\}$, we
conclude that $\mu \in \semp{Q^\star}{\DM(\R,\Sigma,I)}$ by definition
of $Q^\star$ and the fact that $\mu = \tr(t)$.

\medskip

Second, we show that $\semp{Q^\star}{\DM(\R,\Sigma,I)} \subseteq
\tr(\semr{Q})$. Assume that $\mu \in
\semp{Q^\star}{\DM(\R,\Sigma,I)}$. Without loss of generality, assume
that $\dom(\mu) = \{?A_1, \dots, ?A_k\}$, where $0 \leq k \leq
\ell$. Then by definition of $Q^\star$, we have that there exists an IRI
$r_{{\it id}}$ such that $\DM(\R,\Sigma,I)$ contains triples $(r_{{\it
id}}, \rdftype, r)$ and $(r_{{\it id}}, a_i, \mu(?A_i))$, for every $i
\in \{1, \ldots, k\}$, and it does not contain a triple of the
form $(r_{{\it id}}, a_j, v_j)$, for every $j \in \{k+1, \ldots,
\ell\}$. Given the definition of $\DM(\R,\Sigma,I)$ and the fact that
$\IsBinRel(R)$ does not hold, we conclude that there exists a tuple $t
\in R^I$ such that: (1) the IRI assigned by $\DM$ to $t$ is $r_{{\it
id}}$, (2) $t.A_i = \mu(?A_i)$ for every $i \in \{1, \ldots, k\}$, and
(3) $t.A_j = \nil$ for every $j \in \{k+1, \ldots, \ell\}$. Thus,
given that $\tr(t) = \mu$ and $t \in R^I$, we conclude that $\mu \in
\tr(\semr{Q})$ (recall that $\semr{Q} = R^I$).


\item {\bf Binary relation:} Assume that $Q$ is the identity
relational algebra query $R$, where $R$ is a binary relation according
to the definition given in Section \ref{sec-ro}. Moreover, assume that
$\att(R) = \{A_1, A_2\}$, where $A_1$ is a foreign key referencing the
attribute $B$ of a relation $S$, and $A_2$ is a foreign key
referencing the attribute $C$ of a relation $T$. Finally, assume
that $\nu(R) = r$, $\nu(B,S) = b$  and $\nu(C,T) = c$, and define
$Q^\star$ as the following SPARQL 1.1 query:
\begin{eqnarray*}
Q^\star & = & \sel{\{?A_1, ?A_2\}}{((?T_1, r, ?T_2) \andp (?T_1, b,
?A_1) \andp (?T_2, c, ?A_2))}.
\end{eqnarray*}
Next we prove that $\tr(\semr{Q}) = \semp{Q^\star}{\DM(\R,\Sigma,I)}$. 

First, we show that $\tr(\semr{Q}) \subseteq
\semp{Q^\star}{\DM(\R,\Sigma,I)}$. Assume that $\mu \in
\tr(\semr{Q})$. Then there exists a tuple $t \in \semr{Q}$ such that
$\tr(t) = \mu$ and, hence, $t \in R^I$. Given the definition of
mapping $\DM$, we have that all the following hold: $\BinRel(R, A_1, A_2, S, B, T, C)$,
$\PK(A_1, A_2, R)$, $\FK_1(A_1, R, B, S)$, $\FK_1(A_2, R, C, T)$,
$\Class(S)$, $\DTP(B,S)$, $\Class(T)$, $\DTP(C,T)$, $\Rel(S)$,
$\Attr(B,S)$, $\Rel(T)$ and $\Attr(C,T)$. From this, we conclude that
there exist tuples $t_1 \in S^I$, $t_2 \in T^I$ such that $t.A_1 =
t_1.B \neq
\nil$ and $t.A_2 = t_2.C \neq \nil$, and we also conclude that the
following triples are included in $\DM(\R,\Sigma,I)$:
\begin{itemize}
\item $(s_{{\it id}}, r, t_{{\it id}})$ where $s_{{\it id}}$ is
the tuple id for tuple $t_1$ and $t_{{\it id}}$ is the tuple id for tuple
$t_2$, 

\item $(s_{{\it id}}, b, v_1)$, where $v_1$ is the value of attribute
$B$ in the tuple $t_1$, that is, $t_1.B = v_1$, 

\item $(t_{{\it id}}, c, v_2)$, where $v_2$ is the value of attribute
$C$ in the tuple $t_2$, that is, $t_2.C = v_2$.
\end{itemize}
Given that $t.A_1 = t_1.B = v_1$, $t.A_2 = t_2.C = v_2$ and $\tr(t) =
\mu$, we conclude by definition of $Q^\star$ that $\mu \in
\semp{Q^\star}{\DM(\R,\Sigma,I)}$.  

\medskip

Second, we show that $\semp{Q^\star}{\DM(\R,\Sigma,I)} \subseteq
\tr(\semr{Q})$. Assume that $\mu \in
\semp{Q^\star}{\DM(\R,\Sigma,I)}$, which implies that $\dom(\mu) = \{?A_1,
?A_2\}$. By definition of $Q^\star$, we have that there exist IRIs
$s_{{\it id}}$, $t_{{\it id}}$ such that the following triples are in
$\DM(\R,\Sigma,I)$: $(s_{{\it id}}, r, t_{{\it id}})$, $(s_{{\it id}},
b, \mu(?A_1))$ and $(t_{{\it id}}, c, \mu(?A_2))$. Hence, by
definition of $\DM$, we have that there exist tuples $t_1 \in S^I$,
$t_2 \in T^I$ such that: (1) $s_{{\it id}}$ is the IRI assigned to
$t_1$ by $\DM$, (2) $t_1.B = \mu(?A_1)$, (3) $t_{{\it id}}$ is the IRI
assigned to $t_2$ by $\DM$, and (4) $t_2.C = \mu(?A_2)$. Moreover, we
also have by definition of $\DM$ that the following holds: $\BinRel(R, A_1, A_2, S, B, T, C)$, 
$\FK_1(A_1, R, B, S)$ and $\FK_1(A_2, R, C, T)$. Hence, there
exists tuple $t \in R^I$ such that $t.A_1 = t_1.B = \mu(?A_1)$ and
$t.A_2 = t_2.C = \mu(?A_2)$. Therefore, given that $\mu = \tr(t)$
(since $\att(R) = \{A_1, A_2\}$ and $\dom(\mu) = \{?A_1, ?A_2\}$) and
$t \in \semr{Q}$ (since $\semr{Q} = R^I$), we conclude that $\mu \in
\tr(\semr{Q})$.

\item Third, assume that $Q = \nil_A$, and let $Q^\star$ be the
SPARQL query $\egp$. We have that $\semr{Q} = \{t\}$, where $t$ is a
tuple with domain $\{A\}$ such that $t.A = \nil$. Moreover, we have
that $\semp{Q^\star}{\DM(\R,\Sigma,I)} = \{\mu_\emptyset\}$ since
$\DM(\R,\Sigma,I)$ is a nonempty RDF graph. Thus,
given that $\tr(t) = \mu_\emptyset$, we conclude that $\tr(\semr{Q}) =
\semp{Q^\star}{\DM(\R,\Sigma,I)}$.
\end{itemize}

\medskip

\noindent
{\bf Inductive Step}: Assume that the theorem holds for relational
algebra queries $\Qone$ and $\Qtwo$. That is, there exists SPARQL
queries $\Qstarone$ and $\Qstartwo$ such that:
\begin{eqnarray*}
\tr(\semr{\Qone}) & = & \semp{\Qstarone}{\DM(\R,\Sigma,I)},\\  
\tr(\semr{\Qtwo}) & = & \semp{\Qstartwo}{\DM(\R,\Sigma,I)}.
\end{eqnarray*}
To continue with the proof, we need to consider the following
operators: selection ($\sigma$), projection ($\pi$), rename
($\delta$), join ($\nj$), union ($\cup$) and difference
($\smallsetminus$).

\begin{itemize}
\item {\bf Selection:} We need to consider four cases.

\begin{itemize}

\item {\bf Case 1}. Assume that $Q = \sigma_{A_1 = a}(\Qone)$, and
$Q^\star$ = $(\Qstarone \vc (?A_1 = a))$.
Next we prove that $\tr(\semr{Q})  =
\semp{Q^\star}{\DM(\R,\Sigma,I)}$. 

First, we show that $\tr(\semr{Q})
\subseteq \semp{Q^\star}{\DM(\R,\Sigma,I)}$. Assume that $\mu \in
\tr(\semr{Q})$. Then there exists a tuple $t \in \semr{Q}$ such that
$\tr(t) = \mu$. Thus, we have that $t \in \semr{\Qone}$ and $t.A_1 =
a$. By definition of $\tr$, we know that $t.A_1 = \mu(?A_1)$, from
which we conclude that $\mu(?A_1) = a$ given that $t.A_1 =
a$. Therefore, $\mu
\models (?A_1 = a)$, from which we conclude that $\mu \in
\semp{Q^\star}{\DM(\R,\Sigma,I)}$ since $\mu = \tr(t)$ and $\tr(t) \in
\semp{\Qstarone}{\DM(\R,\Sigma,I)}$ by induction hypothesis.

Second, we show that $\semp{Q^\star}{\DM(\R,\Sigma,I)} \subseteq
\tr(\semr{Q})$. Assume that  $\mu \in \semp{Q^\star}{\DM(\R,\Sigma,I)}$. Then 
$\mu \in \semp{\Qstarone}{\DM(\R,\Sigma,I)}$ and $\mu \models (?A_1 =
a)$, that is, $\mu(?A_1) = a$. By induction hypothesis, we have that
$\mu \in \tr(\semr{Q_1})$, and, hence, there exists a tuple $t \in
\semr{\Qone}$ such that $\tr(t) = \mu$. By definition of $\tr$, we
know that $t.A_1 = \mu(?A_1)$, from which we conclude that $t.A_1 = a$
given that $\mu(?A_1) = a$. Given that $t \in \semr{\Qone}$ and $t.A_1
= a$, we have that $t \in \semr{Q}$. Therefore, we conclude that $\mu
\in \tr(\semr{Q})$ since $\tr(t) = \mu$.
 
\item {\bf Case 2}. Assume that $Q = \sigma_{A_1 \neq a}(\Qone)$, and
$Q^\star$ = $(\Qstarone \vc (\neg(?A_1 = a) \wedge \bound(?A_1)))$.
Next we prove that $\tr(\semr{Q}) =
\semp{Q^\star}{\DM(\R,\Sigma,I)}$. 

First, we show that $\tr(\semr{Q})
\subseteq \semp{Q^\star}{\DM(\R,\Sigma,I)}$. Assume that $\mu \in
\tr(\semr{Q})$. Then there exists a tuple $t \in \semr{Q}$ such that
$\tr(t) = \mu$. Given that $t \in \semr{Q}$, we have by the definition
of the semantics of relational algebra that $t \in \semr{\Qone}$,
$t.A_1 \neq a$ and $t.A_1 \neq \nil$. Thus, by definition of $\tr$ we
have that $t.A_1 = \mu(?A_1)$ and $\mu(?A_1) \neq a$. Hence, we have
that $\mu \models (\neg(?A_1 = a) \wedge \bound(?A_1))$, from which we
conclude that $\mu \in \semp{Q^\star}{\DM(\R,\Sigma,I)}$ since $\mu =
\tr(t)$ and $\tr(t) \in \semp{\Qstarone}{\DM(\R,\Sigma,I)}$ by
induction hypothesis. 

Second, we show that $\semp{Q^\star}{\DM(\R,\Sigma,I)} \subseteq
\tr(\semr{Q})$. Assume that  $\mu \in \semp{Q^\star}{\DM(\R,\Sigma,I)}$. Then 
$\mu \in \semp{\Qstarone}{\DM(\R,\Sigma,I)}$ and $\mu \models
(\neg(?A_1 = a) \wedge \bound(?A_1))$, that is, $?A_1 \in \dom(\mu)$
and $\mu(?A_1) \neq a$.  By induction hypothesis we have that $\mu \in
\tr(\semr{Q_1})$ and, hence, there exists a tuple $t \in \semr{\Qone}$
such that $\tr(t) = \mu$. Given that $?A_1 \in \dom(\mu)$ and
$\mu(?A_1) \neq a$, it holds that $t.A_1 \neq \nil$ and $t.A_1 \neq
a$. Thus, we have that $t \in \semr{Q}$, from which we conclude that
$\mu \in \tr(\semr{Q})$ since $\mu = \tr(t)$.
 
\item {\bf Case 3}. Assume that $Q = \sigma_{\isn(A_1)}(\Qone)$, and
$Q^\star = (\Qstarone \vc(\neg\bound(?A_1)))$.  Next we prove that
$\tr(\semr{Q}) =
\semp{Q^\star}{\DM(\R,\Sigma,I)}$. 

First, we show that $\tr(\semr{Q})
\subseteq \semp{Q^\star}{\DM(\R,\Sigma,I)}$. Assume that $\mu \in
\tr(\semr{Q})$. Then there exists a tuple $t \in \semr{Q}$ such that
$\tr(t) = \mu$. Given that $t \in \semr{Q}$, we have that $t
\in \semr{\Qone}$ and $t.A_1 = \nil$. Thus, we conclude by definition
of $\tr$ that $?A_1 \not\in \dom(\mu)$ and, hence, $\mu \models \neg
\bound(?A_1)$. Therefore, we have that $\mu
\in \semp{Q^\star}{\DM(\R,\Sigma,I)}$ given that $\mu = \tr(t)$ and
$\tr(t) \in \semp{\Qstarone}{\DM(\R,\Sigma,I)}$ by induction
hypothesis.

Second, we show that $\semp{Q^\star}{\DM(\R,\Sigma,I)} \subseteq
\tr(\semr{Q})$. Assume that  $\mu \in \semp{Q^\star}{\DM(\R,\Sigma,I)}$. Then 
$\mu \in \semp{\Qstarone}{\DM(\R,\Sigma,I)}$ and $\mu \models
(\neg\bound(?A_1))$, that is, $?A_1 \not\in \dom(\mu)$.  By induction
hypothesis we have that $\mu \in \tr(\semr{Q_1})$, from which we
conclude that there exists a tuple $t \in \semr{\Qone}$ such that
$\tr(t) = \mu$. By definition of $\tr$ and given that $?A_1 \not\in
\dom(\mu)$, we have that $t.A_1 = \nil$ and, hence, $t \in
\semr{Q}$. Therefore, we conclude that $\mu \in \tr(\semr{Q})$ since
$\mu = \tr(t)$.

\item {\bf Case 4}. Assume that $Q = \sigma_{\isnn(A_1)}(\Qone)$, and
$Q^\star = (\Qstarone \vc(\bound(?A_1)))$.  Next we prove that
$\tr(\semr{Q}) =
\semp{Q^\star}{\DM(\R,\Sigma,I)}$. 

First, we show that $\tr(\semr{Q})
\subseteq \semp{Q^\star}{\DM(\R,\Sigma,I)}$. Assume that $\mu \in
\tr(\semr{Q})$. Then there exists a tuple $t \in \semr{Q}$ such that
$\tr(t) = \mu$. Given that $t \in \semr{Q}$, we have that $t \in
\semr{\Qone}$ and $t.A_1 \neq \nil$. Thus, by definition of $\tr$ we
have that $?A_1 \in \dom(\mu)$ and $\mu(?A_1) = t.A_1$ and, hence,
$\mu \models \bound(?A_1)$. Therefore, we conclude that $\mu \in
\semp{Q^\star}{\DM(\R,\Sigma,I)}$ given that $\mu = \tr(t)$ and
$\tr(t) \in \semp{\Qstarone}{\DM(\R,\Sigma,I)}$.

Second, we show that $\semp{Q^\star}{\DM(\R,\Sigma,I)} \subseteq
\tr(\semr{Q})$. Assume that  $\mu \in
\semp{Q^\star}{\DM(\R,\Sigma,I)}$. Then  $\mu \in
\semp{\Qstarone}{\DM(\R,\Sigma,I)}$ and  $\mu \models \bound(?A_1)$,
that is, $?A_1 \in \dom(\mu)$. By induction hypothesis we have that
there exists a tuple $t \in \semr{\Qone}$ such that $\tr(t) =
\mu$. Thus, by definition of $\tr$ we have that $t.A_1 = \mu(?A_1)$,
which implies that $t.A_1 \neq \nil$. Therefore, we have that $t \in
\semr{Q}$ and, hence, $\mu \in \tr(\semr{Q})$ since $\mu =\tr(t)$.
\end{itemize}


\item {\bf Projection:} Assume that $Q = \pi_{\{A_1,  \dots,
A_\ell\}}(\Qone)$, and $Q^\star = \sel{\{?A_1, \ldots,
?A_\ell\}}{\Qstarone}$. Next we prove that $\tr(\semr{Q}) =
\semp{Q^\star}{\DM(\R,\Sigma,I)}$. 

First, we show that $\tr(\semr{Q}) \subseteq
\semp{Q^\star}{\DM(\R,\Sigma,I)}$. Assume that $\mu \in
\tr(\semr{Q})$. Then there exists a tuple $t \in \semr{Q}$ such that
$\tr(t) = \mu$. Given that $t \in \semr{Q}$, there exists a tuple $t'
\in \semr{\Qone}$ such that for every $A \in \att(Q): t.A =
t'.A$. Without loss of generality, assume that: (1) $\att(Q) = \{A_1,
\dots, A_k, A_{k+1},
\dots , A_\ell\}$, (2) $t.A_i \neq \nil$ for every $i \in \{1, \ldots,
k\}$, and (3) $t.A_j = \nil$ for every $j \in \{k+1, \ldots,
\ell\}$. By definition of $\tr$, we have that $t.A_i = \mu(?A_i)$ for
every $i \in \{1, \ldots, k\}$, and that $\dom(\mu) = \{?A_1, \dots,
?A_k\}$. Given that $t' \in \semr{\Qone}$, we have for $\mu' =
\tr(t')$ that: (1) $\mu' \in \tr(\semr{\Qone})$, (2) $\dom(\mu)
\subseteq \dom(\mu')$, (3) $\dom(\mu) = (\{?A_1, \ldots, ?A_\ell\}
\cap \dom(\mu'))$, and (4) $t.A_i = t'.A_i  = \mu(?A_i) = \mu'(?A_i)$
for every $i \in \{1, \ldots, k\}$. Thus, we have in particular that: 
\begin{equation}
\label{eq-res}
\tag{\dag}
\mu \ = \ \res{\mu'}{\{?A_1, \ldots, ?A_\ell\}}.
\end{equation}
By induction hypothesis we have that $\mu' \in
\semp{\Qstarone}{\DM(\R,\Sigma,I)}$, from which we conclude that
$\res{\mu'}{\{?A_1, \ldots, ?A_\ell\}} \in
\semp{Q^\star}{\DM(\R,\Sigma,I)}$. Thus, we conclude from
\eqref{eq-res} that $\mu \in \semp{Q^\star}{\DM(\R,\Sigma,I)}$.

Second, we show that $\semp{Q^\star}{\DM(\R,\Sigma,I)} \subseteq
\tr(\semr{Q})$. Assume that $\mu \in
\semp{Q^\star}{\DM(\R,\Sigma,I)}$. Then there exists a mapping $\mu' \in
\semp{\Qstarone}{\DM(\R,\Sigma,I)}$ such that $\mu = 
\res{\mu'}{\{?A_1, \ldots, ?A_\ell\}}$. By induction hypothesis, we
have that $\mu' \in \tr(\semr{\Qone})$, from we conclude that there
exists a tuple $t' \in \semr{\Qone}$ such that $\tr(t') =\mu'$. Let
$t$ be a tuple with domain $\{A_1, \ldots, A_\ell\}$ such that $t.A_i
= t'.A_i$ for every $i \in \{1, \ldots,
\ell\}$. Then, given that $t' \in \semr{\Qone}$, we have that $t \in
\semr{Q}$, and given that $\mu' = \tr(t')$ and $\mu =
\res{\mu'}{\{?A_1, \ldots, ?A_\ell\}}$, we have that $\mu =
\tr(t)$. Therefore, we conclude that $\mu \in \tr(\semr{Q})$. 


\item {\bf Rename:} Assume that $\att(Q) = \{A_1, \ldots, A_\ell\}$
and $Q = \delta_{A_1 \to B_1}(\Qone)$, and let $Q^\star =
\sel{\{?A_1 \as ?B_1, ?A_2, \dots , ?A_\ell\}}{\Qstarone}$. Next we
prove that $\tr(\semr{Q}) = \semp{Q^\star}{\DM(\R,\Sigma,I)}$.

First, we show that $\tr(\semr{Q}) \subseteq
\semp{Q^\star}{\DM(\R,\Sigma,I)}$. Assume that $\mu \in
\tr(\semr{Q})$. Then there exists a tuple $t \in \semr{Q}$ such that
$\tr(t) = \mu$. Given that $t \in \semr{Q}$, there exists a tuple $t'
\in \semr{\Qone}$ such that $t.B_1 = t'.A_1$ and $t.A_i = t'.A_i$ for
every $i \in \{2, \ldots, \ell\}$. Without loss of generality, assume
that there exists $k \in \{1, \ldots, \ell\}$ such that: (1) $t.A_i
\neq \nil$ for every $i \in \{2, \ldots, k\}$, and (2) $t.A_j = \nil$
for every $j \in \{k+1, \ldots, \ell\}$. To finish the proof, we
consider two cases.
\begin{itemize}
\item Assume that $t.B_1 \neq \nil$. Then it follows from conditions
(1), (2) and definition of $\tr$ that $\mu(?A_1) = t.B_1 = t'.A_1$,
$\mu(?A_i) = t.A_i = t'.A_i$ for every $i \in \{2, \ldots, k\}$ and
$\dom(\mu) = \{?A_1, ?A_2, \dots, ?A_k\}$. Let $\mu' = \tr(t')$. Then
by definition of $\tr$, we have that $\ren{?A_1 \to ?B_1}(\mu') =
\mu$. Moreover, given that $\mu' = \tr(t')$ and $t' \in \semr{Q_1}$,
we conclude that $\mu' \in \tr(\semr{Q_1})$ and, hence, $\mu' \in
\semp{Q_1^\star}{\DM(\R,\Sigma,I)}$ by induction hypothesis. Thus, we
have that $\ren{?A_1 \to 
?B_1}(\res{\mu'}{\{?A_1, \dots, ?A_\ell\}}) \in
\semp{Q^\star}{\DM(\R,\Sigma,I)}$, from which we conclude that $\mu
\in \semp{Q^\star}{\DM(\R,\Sigma,I)}$ since $\res{\mu'}{\{A_1, \ldots,
A_\ell\}} = \mu'$ and $\ren{?A_1 \to ?B_1}(\mu') = \mu$.

\item Assume that $t.B_1 = \nil$. Then it follows from conditions
(1), (2) and definition of $\tr$ that $\mu(?A_i) = t.A_i = t'.A_i$ for
every $i \in \{2,
\ldots, k\}$ and $\dom(\mu) = \{?A_2, ?A_2, \dots, ?A_k\}$. Let $\mu'
= \tr(t')$. Then by definition of $\tr$, we have that $\ren{?A_1 \to
?B_1}(\mu') = \mu$. Moreover, given that $\mu' = \tr(t')$ and $t' \in
\semr{Q_1}$, we conclude that $\mu' \in \tr(\semr{Q_1})$ and, hence,
$\mu' \in \semp{Q_1^\star}{\DM(\R,\Sigma,I)}$ by induction
hypothesis. Thus, we have that $\ren{?A_1 \to ?B_1}(\res{\mu'}{\{?A_1,
\dots, ?A_\ell\}}) \in
\semp{Q^\star}{\DM(\R,\Sigma,I)}$, from which we conclude that $\mu
\in \semp{Q^\star}{\DM(\R,\Sigma,I)}$ since $\res{\mu'}{\{A_1, \ldots,
A_\ell\}} = \mu'$ and $\ren{?A_1 \to ?B_1}(\mu') = \mu$.
\end{itemize}

Second, we show that $\semp{Q^\star}{\DM(\R,\Sigma,I)} \subseteq
\tr(\semr{Q})$. Assume that $\mu \in
\semp{Q^\star}{\DM(\R,\Sigma,I)}$. Then there exists a mapping $\mu'
\in \semp{\Qstarone}{\DM(\R,\Sigma,I)}$ such that $\mu = \ren{?A_1 \to
?B_1}(\res{\mu'}{\{?A_1, \dots, ?A_\ell\}}) $. By induction
hypothesis, we have that $\mu' \in \tr(\semr{\Qone})$, from which we
conclude that there exists a tuple $t' \in \semr{\Qone}$ such that
$\tr(t') = \mu'$. Let $t$ be a tuple with domain $\{B_1, A_2, \ldots,
A_\ell\}$ such that $t.B_1 = t'.A_1$ and $t.A_i = t'.A_i$ for every $i
\in \{2, \ldots, \ell\}$. Then we have that $t \in \semr{Q}$. Given
that $\mu' = \tr(t')$ and $\mu = \ren{?A_1 \to
?B_1}(\res{\mu'}{\{?A_1, \dots, ?A_\ell\}}) $, we have that $\mu =
\tr(t)$. Therefore, we conclude that $\mu \in \tr(\semr{Q})$. 


\item {\bf Join:} Assume that $Q = (Q_1 \nj Q_2)$, where $(\att(Q_1)
\cap \att(Q_2)) = \{A_1, \ldots, A_\ell\}$, and let 
\begin{multline*}
Q^\star \ = \ \bigg(\bigg(Q_1^\star \vc (\bound(?A_1) \wedge \cdots
\wedge \bound(?A_\ell))\bigg) \ \andp \bigg(Q_2^\star \vc (\bound(?A_1)
\wedge \cdots \wedge \bound(?A_\ell))\bigg)\bigg).
\end{multline*}
Next we prove that $\tr(\semr{Q}) =
\semp{Q^\star}{\DM(\R,\Sigma,I)}$. 

First, we show that $\tr(\semr{Q}) \subseteq
\semp{Q^\star}{\DM(\R,\Sigma,I)}$. Assume that $\mu \in
\tr(\semr{Q})$. Then there exists a tuple $t$ such that $\mu = \tr(t)$
and $t \in \semr{Q}$. Thus, we have that there exist tuples $t_1 \in
\semr{Q_1}$ and $t_2 \in \semr{Q_2}$ such that: (1) $t.A_i = t_1.A_i =
t_2.A_i \neq \nil$ for every $i \in \{1, \ldots, \ell\}$, (2) $t.A =
t_1.A$ for every $A \in (\att(Q_1) \smallsetminus \att(Q_2))$, and (3)
$t.A = t_2.A$ for every $A \in (\att(Q_2) \smallsetminus
\att(Q_1))$. Let $\mu_1 = \tr(t_1)$ and $\mu_2 = \tr(t_2)$. By
induction hypothesis and given that $\mu_1 \in \tr(\semr{Q_1})$
and $\mu_2 \in \tr(\semr{Q_2})$, we have that $\mu_1 \in
\semp{Q_1^\star}{\DM(\R,\Sigma,I)}$ and $\mu_2 \in
\semp{Q_2^\star}{\DM(\R,\Sigma,I)}$. Hence, from condition (1) and
definition of $\tr$, we conclude that: 
\begin{eqnarray*}
\mu_1 & \in & \semp{(Q_1^\star \vc (\bound(?A_1) \wedge \cdots \wedge
\bound(?A_\ell)))}{\DM(\R,\Sigma,I)},\\
\mu_2 & \in & \semp{(Q_2^\star \vc (\bound(?A_1) \wedge \cdots \wedge
\bound(?A_\ell)))}{\DM(\R,\Sigma,I)}.
\end{eqnarray*}
Thus, given that $\mu = \mu_1 \cup \mu_2$ by conditions (1), (2), (3)
and definition of $\tr$, we conclude that $\mu \in
\semp{Q^\star}{\DM(\R,\Sigma,I)}$.

\begin{sloppypar}
Second, we show that $\semp{Q^\star}{\DM(\R,\Sigma,I)} \subseteq
\tr(\semr{Q})$.  Assume that $\mu \in
\semp{Q^\star}{\DM(\R,\Sigma,I)}$. Then there exist mappings $\mu_1$,
$\mu_2$ such that: (1) $\mu = \mu_1 \cup \mu_2$, (2) $\mu_1 \in
\semp{(Q_1^\star \vc (\bound(?A_1) \wedge \cdots$ $\cdots \wedge
\bound(?A_\ell)))}{\DM(\R,\Sigma,I)}$, and (3) $\mu_2 \in
\semp{(Q_2^\star \vc (\bound(?A_1) \wedge \cdots \wedge
\bound(?A_\ell)))}{\DM(\R,\Sigma,I)}$. By induction hypothesis, we
have that $\mu_1 \in \tr(\semr{Q_1})$ and $\mu_2 \in
\tr(\semr{Q_2})$. Thus, there exist tuples $t_1 \in \semr{Q_1}$, $t_2
\in \semr{Q_2}$ such that $\mu_1 = \tr(t_1)$ and $\mu_2 =
\tr(t_2)$. From conditions (1), (2), (3) and definition of $\tr$, we
have that $t_1.A_i = t_2.A_i = \mu(?A_i) \neq \nil$ for every $i \in
\{1, \ldots, \ell\}$. Thus, given that $(\att(Q_1) \cap \att(Q_2)) =
\{A_1, \ldots, A_\ell\}$, we have that $t \in \semr{Q}$, where $t :
(\att(Q_1)
\cup \att(Q_2)) \to (\D \cup \{\nil\})$ such that: (4) $t.A_i =
t_1.A_i = t_2.A_i$ for every $i \in \{1, \ldots, \ell\}$, (5) $t.A =
t_1.A$ for every $A \in (\att(Q_1) \smallsetminus \att(Q_2))$, and (6)
$t.A = t_2.A$ for every $A \in (\att(Q_2) \smallsetminus
\att(Q_1))$. Hence, we conclude that $\mu \in \tr(\semr{Q})$, given
that $\mu = \tr(t)$ by definition of $t$, definition of $\tr$ and
conditions (1), (2) and (3).
\end{sloppypar}


\item {\bf Union:} Assume that $Q = (Q_1 \cup Q_2)$ and
$Q^\star = (Q_1^\star \uni Q_2^\star)$. 
Next we prove that $\tr(\semr{Q}) =
\semp{Q^\star}{\DM(\R,\Sigma,I)}$. 

First, we show that $\tr(\semr{Q})
\subseteq \semp{Q^\star}{\DM(\R,\Sigma,I)}$. Assume that $\mu \in
\tr(\semr{Q})$. Then there exists a tuple $t \in \semr{Q}$ such that
$\tr(t) = \mu$. Thus, we have that $t \in \semr{Q_1}$ or $t \in
\semr{Q_2}$. Without loss of generality, assume that $t \in
\semr{Q_1}$. Then we have that $\tr(t) \in \tr(\semr{Q_1})$ and, hence,
$\tr(t) \in \semp{Q_1^\star}{\DM(\R,\Sigma,I)}$ by induction
hypothesis. Therefore, $\mu \in \semp{Q_1^\star}{\DM(\R,\Sigma,I)}$
since $\tr(t) = \mu$, from which we conclude that $\mu \in
\semp{Q^\star}{\DM(\R,\Sigma,I)}$. 

Second, we show that $\semp{Q^\star}{\DM(\R,\Sigma,I)} \subseteq
\tr(\semr{Q})$. Assume that  $\mu \in \semp{Q^\star}{\DM(\R,\Sigma,I)}$. Then 
$\mu \in \semp{Q_1^\star}{\DM(\R,\Sigma,I)}$ or $\mu \in
\semp{Q_2^\star}{\DM(\R,\Sigma,I)}$. Without loss of generality,
assume that $\mu \in \semp{Q_1^\star}{\DM(\R,\Sigma,I)}$. Then, by induction
hypothesis, we have that $\mu \in \tr(\semr{Q_1})$, and, hence, there
exists a tuple $t \in \semr{Q_1}$ such that $\tr(t) = \mu$. Therefore,
we conclude that $t \in \semr{(Q_1 \cup Q_2)}$, from which we deduce
that $\mu \in \tr(\semr{Q})$.

\item {\bf Difference:} Assume that $Q = (Q_1 \smallsetminus Q_2)$,
and that $\att(Q_1) = \att(Q_2) = \{A_1, \ldots, A_\ell\}$. Then for
every (not necessarily nonempty) set $\X = \{i_1, i_2, \ldots, i_p\}$
such that $1 \leq i_1 < i_2 < \ldots < i_p \leq \ell$, define $R_\X$
as the following filter condition:
\begin{multline*}
\bigg(\bound(?A_{i_1}) \wedge \bound(?A_{i_2}) \wedge \cdots \wedge
\bound(?A_{i_p}) \ \wedge\\ \neg \bound (?A_{j_1}) \wedge \neg \bound
(?A_{j_2}) \wedge \cdots \wedge \neg \bound (?A_{j_q})\bigg),
\end{multline*}
where $1 \leq j_1 < j_2 < \cdots < j_q \leq \ell$ and $\{j_1, j_2,
\cdots, j_q\} = (\{1, \ldots, \ell\} \smallsetminus \{i_1, i_2,
\ldots, i_p\})$. That is, condition $R_\X$ indicates that every
variables $?A_i$ with $i \in \X$ is bound, while every variable $?A_j$
with $j \in (\{1, \ldots, \ell\} \smallsetminus \X)$ is not
bound. Moreover, for every $\X \neq \emptyset$ define SPARQL graph
pattern $P_\X$ as follows: 
\begin{eqnarray*}
P_\X & = & ((Q^\star_1 \vc R_\X) \minus (Q^\star_2 \vc R_\X)).
\end{eqnarray*} 
Notice that there are $2^\ell-1$ possible graph patterns $P_\X$ with
$\X \neq \emptyset$. Let $P_1$, $P_2$, $\ldots$, $P_{2^\ell-1}$ be an
enumeration of these graph patterns. Moreover, assuming that $?X$,
$?Y$, $?Z$ are fresh variables, let $P_\emptyset$ be the following
query: 
\begin{eqnarray*}
&\bigg[\bigg[\bigg(Q^\star_1 \vc R_\emptyset\bigg) \opt
\bigg(\bigg(Q^\star_2 \vc R_\emptyset\bigg) \andp
(?X,?Y,?Z)\bigg)\bigg] \vc (\neg \bound(?X))\bigg].&
\end{eqnarray*}
Then graph pattern $Q^\star$ is defined as follows:
\begin{eqnarray*}
Q^\star & = & (P_1 \uni P_2 \uni \cdots \uni P_{2^\ell-1} \uni
P_\emptyset).  
\end{eqnarray*}
Next we show that $\tr(\semr{Q}) =
\semp{Q^\star}{\DM(\R,\Sigma,I)}$. In this proof, we assume, by
considering Lemma \ref{lem-sel}, that for every mapping $\mu$ such
that $\mu \in \semp{Q_1^\star}{\DM(\R,\Sigma,I)}$ or $\mu \in
\semp{Q_2^\star}{\DM(\R,\Sigma,I)}$, it holds that $\dom(\mu)
\subseteq \{?A_1, \ldots, ?A_\ell\}$. 

First, we show that $\tr(\semr{Q})
\subseteq \semp{Q^\star}{\DM(\R,\Sigma,I)}$.  Assume that $\mu
\in \tr(\semr{Q})$. Then there exists a tuple $t \in \semr{Q}$ such
that $\tr(t) = \mu$. Thus, we have that $t \in \semr{Q_1}$ and $t
\not\in \semr{Q_2}$, from which we conclude by considering the
induction hypothesis that $\mu \in \semp{Q_1^\star}{\DM(\R,\Sigma,I)}$
and $\mu \not\in \semp{Q_2^\star}{\DM(\R,\Sigma,I)}$. We consider two
cases to show that this implies that $\mu \in
\semp{Q^\star}{\DM(\R,\Sigma,I)}$. 
\begin{itemize}
\begin{sloppypar}
\item Assume that $\dom(\mu) \neq \emptyset$, and let $\X = \{ i
\in \{1, \ldots, \ell\} \mid ?A_i \in \dom(\mu)\}$. Given that $\mu
\in \semp{Q_1^\star}{\DM(\R,\Sigma,I)}$, we have that $\dom(\mu)
\subseteq \{?A_1, \ldots, ?A_\ell\}$ and, hence, $\X \neq \emptyset$. 
Furthermore, we have that $\mu \models R_\X$ and, hence, $\mu \in
\semp{(Q_1^\star \vc R_\X)}{\DM(\R,\Sigma,I)}$. From this and the fact
that $\mu \not\in \semp{Q_2^\star}{\DM(\R,\Sigma,I)}$, we conclude
that: 
\begin{equation}
\label{eq-minus}
\tag{\ddag}
\mu \ \in \ \semp{((Q_1^\star \vc R_\X) \minus (Q_2^\star \vc
R_\X))}{\DM(\R,\Sigma,I)}. 
\end{equation}
To see why this is the case, assume that \eqref{eq-minus} does not
hold. Then given that $\mu \in \semp{(Q_1^\star \vc
R_\X)}{\DM(\R,\Sigma,I)}$, we conclude by definition of the operator
$\MINUS$ that there exists a mapping $\mu' \in \semp{(Q_2^\star \vc
R_\X)}{\DM(\R,\Sigma,I)}$ such that $\mu \sim \mu'$ and $(\dom(\mu) \cap
\dom(\mu')) \neq \emptyset$. Given that $\mu' \in
\semp{Q^\star_2}{\DM(\R,\Sigma,I)}$,  we have that $\dom(\mu')
\subseteq \{?A_1, \ldots, ?A_\ell\}$. Thus, given that $\mu' \models
R_\X$ and $\dom(\mu) \subseteq \{?A_1, \ldots, ?A_\ell\}$, we conclude
that $\dom(\mu) =
\dom(\mu')$. Therefore, given that $\mu \sim \mu'$, we have that $\mu =
\mu'$, from which we conclude that $\mu \in
\semp{Q_2^\star}{\DM(\R,\Sigma,I)}$, leading to a contradiction.

{From} \eqref{eq-minus} and definition of $Q^\star$, we conclude that
$\mu \in \semp{Q^\star}{\DM(\R,\Sigma,I)}$ since $((Q_1^\star \vc
R_\X) \minus (Q_2^\star \vc R_\X)) = P_i$ for some $i \in \{1, \ldots,
2^{\ell}-1\}$ (recall that $\X \neq \emptyset$).

\item Assume that $\dom(\mu) = \emptyset$. Then we have that $\mu
\models R_\emptyset$ and, hence, $\mu \in \semp{(Q_1^\star \vc
R_\emptyset)}{\DM(\R,\Sigma,I)}$. From this and the fact 
that $\mu \not\in \semp{Q_2^\star}{\DM(\R,\Sigma,I)}$, we conclude
that:
\begin{align}
\notag
\mu \ \in \ \semp{\bigg[\bigg[&\bigg(Q^\star_1 \vc R_\emptyset\bigg) \opt\\
\label{eq-minus-empty}
\tag{$*$}
&\bigg(\bigg(Q^\star_2 \vc R_\emptyset\bigg) \andp
(?X,?Y,?Z)\bigg)\bigg] \vc (\neg \bound(?X))\bigg]}{\DM(\R,\Sigma,I)}.
\end{align}
To see why this is the case, assume that \eqref{eq-minus-empty} does
not hold. Then given that $\mu \in \semp{(Q_1^\star \vc
R_\emptyset)}{\DM(\R,\Sigma,I)}$ and $\dom(\mu) = \emptyset$, we have
that there exists a mapping $\mu' \in \semp{((Q_2^\star \vc
R_\emptyset) \andp (?X,?Y,?Z))}{\DM(\R,\Sigma,I)}$ such that $?X \in
\dom(\mu')$. Thus, there exist mappings $\mu_1 \in \semp{(Q_2^\star
\vc R_\emptyset)}{\DM(\R,\Sigma,I)}$ and $\mu_2 \in
\semp{(?X,?Y,?Z)}{\DM(\R,\Sigma,I)}$ such that $\mu' = \mu_1 \cup
\mu_2$. Given that  
$\mu_1 \in \semp{(Q_2^\star \vc R_\emptyset)}{\DM(\R,\Sigma,I)}$, we
have that $\mu_1 \in \semp{Q_2^\star}{\DM(\R,\Sigma,I)}$ and $\mu_1
\models R_\emptyset$. Thus, we have that $\dom(\mu_1) \subseteq
\{?A_1, \ldots, ?A_\ell\}$, from which we conclude that $\dom(\mu_1) =
\emptyset$ (since $\mu_1 \models R_\emptyset$). Therefore, we have
that $\mu = \mu_1$, which implies that $\mu \in
\semp{Q_2^\star}{\DM(\R,\Sigma,I)}$ and leads to a contradiction.

{From} \eqref{eq-minus-empty} and definition of $Q^\star$, we conclude
that $\mu \in \semp{Q^\star}{\DM(\R,\Sigma,I)}$.
\end{sloppypar}
\end{itemize}

Second, we show that $\semp{Q^\star}{\DM(\R,\Sigma,I)} \subseteq
\tr(\semr{Q})$.  Assume that $\mu \in
\semp{Q^\star}{\DM(\R,\Sigma,I)}$. Then we consider two cases to prove
that $\mu \in \tr(\semr{Q})$.
\begin{itemize}
\item Assume that there exists $i \in \{1, \ldots, \ell\}$ such that
$\mu \in \semp{P_i}{\DM(\R,\Sigma,I)}$. Then there exists $\X \neq
\emptyset$ such that $\mu \in \semp{((Q_1^\star \vc R_\X) \minus
(Q_2^\star \vc R_\X))}{\DM(\R,\Sigma,I)}$. Thus, we have that $\mu \in
\semp{Q_1}{\DM(\R,\Sigma,I)}$ and $\mu \models R_\X$, from which we
conclude that $\emptyset \subsetneq \dom(\mu) \subseteq \{?A_1, \ldots,
?A_\ell\}$. From this fact and definition of the $\MINUS$
operator, we obtain that $\mu \not\in
\semp{Q_2^\star}{\DM(\R,\Sigma,I)}$. Hence, by induction hypothesis,
we conclude that $\mu \in \tr(\semr{Q_1})$ and $\mu \not\in
\tr(\semr{Q_2})$. That is, there exists a tuple $t$ such that $\tr(t)
= \mu$, $t \in \semr{Q_1}$ and $t \not\in \semr{Q_2}$, from which we
conclude that $\mu \in \tr(\semr{Q})$.

\item \begin{sloppypar}
Assume that \eqref{eq-minus-empty} holds. First we show that
$\semp{(Q^\star_2 \vc R_\emptyset)}{\DM(\R,\Sigma,I)} =
\emptyset$. For the sake of contradiction, assume that there exists a
mapping $\mu' \in \semp{(Q^\star_2 \vc
R_\emptyset)}{\DM(\R,\Sigma,I)}$. Then given that $\mu' \in
\semp{Q^\star_2}{\DM(\R,\Sigma,I)}$ and $\mu \models R_\emptyset$, we
conclude that $\dom(\mu') = \emptyset$. Given that $\DM(\R,\Sigma,I)$
is a nonempty RDF graph and $\dom(\mu') = \emptyset$, we conclude that
there exists a mapping $\mu'' \in \semp{((Q^\star_2 \vc R_\emptyset)
\andp (?X,?Y,?Z))}{\DM(\R,\Sigma,I)}$ such that $\dom(\mu'') = \{?X,
?Y, ?Z\}$. Thus, given that variables $?X$, $?Y$, $?Z$ are not
mentioned in $(Q^\star_1 \vc R_\emptyset)$, we conclude that $\mu''$
is compatible with every mapping in $\semp{(Q^\star_1 \vc
R_\emptyset)}{\DM(\R,\Sigma,I)}$. Thus, by definition of the $\OPT$
operator, we conclude that $?X$ belongs to the domain of every mapping
in $\semp{((Q^\star_1 \vc R_\emptyset) \opt ((Q^\star_2 \vc
R_\emptyset) \andp (?X,?Y,?Z)))}{\DM(\R,\Sigma,I)}$, which implies
that $\llbracket(((Q^\star_1 \vc R_\emptyset) \opt ((Q^\star_2 \vc
R_\emptyset) \andp (?X,?Y,?Z)))$ $\vc (\neg
\bound(?X)))\rrbracket_{\DM(\R,\Sigma,I)} = \emptyset$. But this leads to a
contradiction, as we assume that \eqref{eq-minus-empty} holds.

Given that \eqref{eq-minus-empty} holds and  $\semp{(Q^\star_2 \vc
R_\emptyset)}{\DM(\R,\Sigma,I)} = \emptyset$, we conclude that $\mu
\in \semp{(Q^\star_1 \vc R_\emptyset)}{\DM(\R,\Sigma,I)}$ and 
$\mu \not\in \semp{(Q_2^\star \vc
R_\emptyset)}{\DM(\R,\Sigma,I)}$. Hence, we have that $\mu \in
\semp{Q^\star_1}{\DM(\R,\Sigma,I)}$ and $\mu \not\in
\semp{Q^\star_2}{\DM(\R,\Sigma,I)}$ and, therefore, we conclude by
induction hypothesis that $\mu \in \tr(\semr{Q_1})$ and $\mu \not\in 
\tr(\semr{Q_2})$. That is, there exists a tuple $t$ such that $\tr(t)
= \mu$, $t \in \semr{Q_1}$ and $t \not\in \semr{Q_2}$, from which we
conclude that $\mu \in \tr(\semr{Q})$.
\end{sloppypar}
\end{itemize}
\end{itemize}

\subsection{Proof of Proposition \ref{prop-dm-sp}}
Assume that we have a relational schema containing a relation with
name {\tt STUDENT} and attributes {\tt SID}, {\tt NAME}, and assume
that the attribute {\tt SID} is the primary key.  Moreover, assume
that this relation has two tuples, $t_1$ and $t_2$ such that
$t_1.\text{\tt SID}= \text{1}$, $t_1.\text{\tt NAME} = \text{John}$
and $t_2.\text{\tt SID} = \text{1}$, $t_2.\text{\tt NAME} =
\text{Peter}$. It is clear that the primary key is violated, therefore
the database is inconsistent. If $\DM$ would be semantics preserving,
then the resulting RDF graph would be inconsistent under OWL
semantics. However, the result of applying $\DM$, returns the
following consistent RDF graph (assuming given a base IRI {\tt
"http://example.edu/db/"} for the mapping):

\vspace{-10pt}

{\small
\begin{align*}
\Triple(&\text{\tt "http://example.edu/db/STUDENT"}, \rdftype, \owlClass)  \\
\Triple(&\text{\tt "http://example.edu/db/STUDENT\#NAME"}, \rdftype, \owlDTP)  \\
\Triple(&\text{\tt "http://example.edu/db/STUDENT\#NAME"}, \domain, \text{\tt "http://example.edu/db/STUDENT"})  \\
\Triple(&\text{\tt "http://example.edu/db/STUDENT\#SID"}, \rdftype, \owlDTP)  \\
\Triple(&\text{\tt "http://example.edu/db/STUDENT\#SID"}, \domain, \text{\tt "http://example.edu/db/STUDENT"})  \\
\Triple(&\text{\tt "http://example.edu/db/STUDENT\#SID=1"}, \text{\tt "http://example.edu/db/STUDENT\#NAME"}, \text{\tt "John"})  \\
\Triple(&\text{\tt "http://example.edu/db/STUDENT\#SID=1"}, \text{\tt "http://example.edu/db/STUDENT\#NAME"}, \text{\tt "Peter"}) \\
\Triple(&\text{\tt "http://example.edu/db/STUDENT\#SID=1"}, \text{\tt "http://example.edu/db/STUDENT\#SID"}, \text{\tt "1"})
\end{align*}
\bw{Therefore,}} $\DM$ is not semantics preserving. \qed

\subsection{Proof of Proposition \ref{prop-dmpk-sp}}
\begin{sloppypar}
It is straightforward to see that given a relational schema $\R$, set
$\Sigma$ of (only) PKs over $\R$ and instance $I$ of $\R$ such that $I
\models \Sigma$, it holds that $\DM_{\text{\it pk}}(\R, \Sigma, I)$ is
consistent under the OWL semantics. Likewise, if $I \not\models
\Sigma$, then by definition of $\DM_{\text{\it pk}}$, the resulting RDF graph
will have an inconsistent triple \Triple(a, \owldf, a), which would
generate an inconsistency under the OWL semantics.
\end{sloppypar}

\subsection{Proof of Theorem \ref{theo-mono-map-no}}
For the sake of contradiction, assume that $\M$ is a monotone and
semantics preserving direct mapping. Then consider a schema $\R$
containing at least two distinct relation names $R_1$, $R_2$, and
consider a set $\Sigma$ of PKs and FKs over $\R$ containing at least
one foreign key from $R_1$ to $R_2$. Then we have that there exist
instances $I_1$, $I_2$ of $\R$ such that $I_1 \subseteq I_2$, $I_1$
does not satisfy $\Sigma$ and $I_2$ does satisfy $\Sigma$. Given that
$\M$ is semantics preserving, we know that $\M(\R, \Sigma, I_2)$ is
consistent under the OWL semantics, while $\M(\R, \Sigma, I_1)$
is not. Given that $\M$ is monotone, we have that $\M(\R, \Sigma, I_1)
\subseteq \M(\R, \Sigma, I_2)$. But then we conclude that $\M(\R,
\Sigma, I_1)$ is also consistent under the OWL semantics, given
that  $\M(\R, \Sigma, I_2)$ is consistent and $\M(\R, \Sigma, I_1)
\subseteq \M(\R, \Sigma, I_2)$, which leads to a contradiction.

\subsection{Proof of Theorem \ref{theo-dm-sp-pk-fk}}
It is straightforward to see that given a relational schema $\R$, set
$\Sigma$ of PKs and FKs over $\R$ and instance $I$ of $\R$ such that
$I \models \Sigma$, it holds that $\DM_{\text{\it pk} + \text{\it
fk}}(\R, \Sigma, I)$ is consistent under the OWL
semantics. Likewise, if $I \not\models \Sigma$, then by definition of
$\DM_{\text{\it pk}+ \text{\it fk}}$, the resulting RDF graph will
contain an inconsistent triple \Triple(a, \owldf, a), which would
generate an inconsistency under the OWL semantics.


\begin{thebibliography}{10}
\smallskip

{\small
\bibitem{OWL2}
{W3C OWL Working Group}. {OWL 2 Web} ontology language document overview.
\newblock {W3C} Recommendation 27 October 2009,
  \verb+http://www.w3.org/TR/owl2-overview/+.
  
\bibitem{D2R}
{D2R Server}. Publishing Relational Databases on the Semantic Web
{\scriptsize \verb+http://www4.wiwiss.fu-berlin.de/bizer/d2r-server/+}.

\bibitem{AHV95}
S.~Abiteboul, R.~Hull, and V.~Vianu.
\newblock {\em Foundations of Databases}.
\newblock Addison-Wesley, 1995.

\bibitem{AG08}
R.~Angles and C.~Gutierrez.
\newblock The expressive power of sparql.
\newblock In {\em ISWC}, pages 114--129, 2008.

\bibitem{APS11}
M.~Arenas, A.~Bertails, E.~Prud'hommeaux, and J.~Sequeda.
\newblock Direct mapping of relational data to {RDF}.
\newblock {W3C} Working Draft 20 September 2011,
  \verb+http://www.w3.org/TR/rdb-direct-mapping/+.

\bibitem{BP11}
A.~Bertails, and E.~Prud'hommeaux.
\newblock Interpreting relational databases in the RDF domain
\newblock In {\em K-CAP}, pages 129--136, 2011.

\bibitem{CGL+07}
D.~Calvanese, G.~D. Giacomo, D.~Lembo, M.~Lenzerini, and R.~Rosati.
\newblock Eql-lite: Effective first-order query processing in description
  logics.
\newblock In {\em IJCAI}, pages 274--279, 2007.

\bibitem{Cer08}
F.~Cerbah.
\newblock Mining the Content of Relational Databases to Learn Ontologies with Deeper Taxonomies
\newblock In {\em Web Intelligence}, pages 553-557, 2008.

\bibitem{DLN+98}
F.~Donini, M.~Lenzerini, D.~Nardi, W.~Nutt, and A.~Schaerf.
\newblock An epistemic operator for description logics.
\newblock {\em Artif. Intell.}, 100(1-2):225--274, 1998.

\bibitem{DNR02}
F.~M. Donini, D.~Nardi, and R.~Rosati.
\newblock Description logics of minimal knowledge and negation as failure.
\newblock {\em ACM TOCL}, 3(2):177--225, 2002.

\bibitem{GM05}
S.~Grimm and B.~Motik.
\newblock Closed world reasoning in the semantic web through epistemic
  operators.
\newblock In {\em OWLED}, 2005.

\bibitem{HS11}
S.~Harris and A.~Seaborne.
\newblock {SPARQL 1.1} query language.
\newblock W3C Working Draft 12 May 2011,
  \verb+http://www.w3.org/TR/sparql11-query/+.

\bibitem{BMZ+07}
B.~He, M.~Patel, Z.~Zhang, and K.~C.-C. Chang.
\newblock Accessing the deep web.
\newblock {\em Commun. ACM}, 50:94--101, May 2007.

\bibitem{MRG11}
A.~Mehdi, S.~Rudolph, and S.~Grimm.
\newblock Epistemic querying of OWL knowledge bases.
\newblock In {\em ESWC (1)}, pages 397--409, 2011.

\bibitem{MHS09}
B.~Motik, I.~Horrocks, and U.~Sattler.
\newblock Bridging the gap between OWL and relational databases.
\newblock {\em J. Web Sem.}, 7(2):74--89, 2009.

\bibitem{PAG09}
J.~P{\'e}rez, M.~Arenas, and C.~Gutierrez.
\newblock {Semantics and complexity of SPARQL}.
\newblock {\em ACM Trans. Database Syst.}, 34(3), 2009.

\bibitem{PS08}
E.~Prud'hommeaux and A.~Seaborne.
\newblock {SPARQL} query language for {RDF}.
\newblock {W3C} Recommendation 15 January 2008,
  \verb+http://www.w3.org/TR/rdf-sparql-query/+.

\bibitem{Rei88}
R.~Reiter.
\newblock On integrity constraints.
\newblock In {\em TARK}, pages 97--111, 1988.


\bibitem{STC+12}
J.~F. Sequeda, S.~H. Tirmizi, O.~Corcho, and D.~P. Miranker.
\newblock Survey of directly mapping sql databases to the semantic web.
\newblock {\em Knowledge Eng. Review}, 26(4): 445-486 (2011)

\bibitem{SFD09}
I.~Seylan, E.~Franconi, and J.~De Bruijn.
\newblock Effective query rewriting with ontologies over DBoxes.
\newblock In {\em IJCAI}, pages 923--929, 2009.


\bibitem{TSB+10}
J.~Tao, E.~Sirin, J.~Bao, and D.~L. McGuinness.
\newblock Integrity constraints in OWL.
\newblock In {\em AAAI}, 2010.

\bibitem{TSM08}
S.~H. Tirmizi, J.~Sequeda, and D.~P. Miranker.
\newblock {Translating SQL Applications to the Semantic Web}.
\newblock In {\em DEXA}, pages 450--464, 2008.
}

\end{thebibliography}

\begin{thebibliography}{10}
\setcounter{enumi}{22}

\medskip

{\small

\bibitem{AG10}
R.~Angles and C.~Gutierrez.
\newblock {SQL} nested queries in {SPARQL}.
\newblock In {\em AMW}, 2010.

\bibitem{AG11}
R.~Angles and C.~Gutierrez.
\newblock Subqueries in {SPARQL}.
\newblock In {\em AMW}, 2011.

}

\end{thebibliography}
\end{document}